\documentclass{IEEEtran}
\onecolumn
\usepackage[pdftex]{graphicx}
\usepackage{amsmath}
\usepackage{amssymb}
\usepackage{amsthm}

\usepackage{cancel}
\usepackage{tikz}
\newcommand{\hc}[1]{%
    \tikz[baseline=(tocancel.base)]{
        \node[inner sep=0pt,outer sep=0pt] (tocancel) {#1};
        \draw[black] (tocancel.south west) -- (tocancel.north east);
    }%
}
\newtheorem{theorem}{Theorem}
\newtheorem{corollary}{Corollary}
\newtheorem{proposition}{Proposition} 
\newtheorem{conjecture}{Conjecture}
\newtheorem{definition}{Definition}
\newtheorem{remark}{Remark}
\newtheorem{lemma}{Lemma }

\newtheorem{notation}{Notation}

\newcommand{\wt}{\widetilde}
\newcommand{\wh}{\widehat}
\newcommand{\lra}{\leftrightarrow}
\newcommand{\mb}{\mathbf}
\newcommand{\mf}{ }
\newcommand{\si}{Sc}
\newcommand{\tr}{Tr}
\newcommand{\mse}{\ensuremath{\textsc{MSE}}}
\newcommand{\RDSI}{\ensuremath{\textsc{RDSI}}}
\newcommand{\RDTr}{\ensuremath{\textsc{RDTr}}}
\newcommand{\RDmse}{\ensuremath{\textsc{RDmse}}}
\newcommand{\ELB}{\ensuremath{\textsc{E}^2\textsc{LB}}}
\newcommand{\DSM}{\ensuremath{\textsc{ELB}}}
\newcommand{\MLB}{\ensuremath{\textsc{MLB}}}
\newcommand{\mLB}{\ensuremath{\textsc{mLB}}}

\DeclareMathOperator{\Tr}{Tr}

\usepackage{graphicx}
\usepackage{latexsym}
\usepackage{cite}
\usepackage{extraipa}
\usepackage{mathrsfs}
\usepackage{stmaryrd}
\usepackage{overpic}
\usepackage{hyperref}
\usepackage{pdfsync}

\usepackage{subcaption}
\begin{document}

\sloppy

\title{Vector Gaussian Rate-Distortion \\ with Variable Side Information} 

\author{Sinem Unal and Aaron B. Wagner~
\thanks{S.~Unal and A.~B.~Wagner are with  Cornell University, School of Electrical \& Computer Engineering, Ithaca, NY 14853 USA (e-mail: su62@cornell.edu, wagner@cornell.edu). This paper was presented in parts at the IEEE Int. Symposium on Information Theory (ISIT), Hawaii, July 2014, the Information Theory and Applications Workshop, Jan./Feb. 2016, and the Conference on Information Sciences and Systems (CISS), Princeton, March 2016.}
}

\maketitle

\begin{abstract}
We consider rate-distortion with two decoders, each
with distinct side information. This problem is well understood 
when the side information at the decoders satisfies a 
certain degradedness condition. We consider cases in which 
this degradedness condition is violated but the source and the 
side information consist of jointly Gaussian vectors. We provide
a hierarchy of four lower bounds on the optimal rate. These
bounds are then used to determine the optimal rate for several
classes of instances.
\end{abstract}

\section{Introduction}

We consider a rate-distortion problem with multiple decoders,
each with potentially different side information, as shown in 
Fig.~\ref{fig:setup}. 
This problem, which is sometimes called the \emph{Heegard-Berger
problem}, is one of the most basic network information
theory problems that is not well understood, and it can
arise in a couple of ways.
First, the side information may represent reconstructions of
the source due to prior transmissions. If the decoders have
received different transmissions in the past, either due to
channel loss or because they are not always listening to the transmitter,
then their side information will be different.
Second, it can be viewed as an instance of a ``compound Wyner-Ziv,''
problem, in which there is in reality only one decoder, but the
encoder is uncertain about its side information. The side
information associated with different decoders in the problem
could then represent the transmitter's view of the set of possible
side information configurations at the decoder. The
transmitter then seeks to construct a message that would work
reasonably well for any of these possible side information
configurations. 
\begin{figure}[t]
  \begin{center}
	 \includegraphics[scale=0.6]{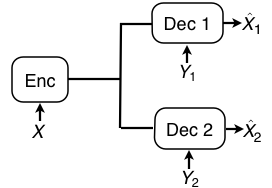}
	\caption{Problem Setup}
   \label{fig:setup}
 \end{center}
\end{figure}

The problem is well understood when the problem is \emph{degraded},
i.e., the side information at one of the decoders is stochastically 
degraded with respect to the other's~\cite{berger}. Recently, the
following class of instances, schemes for which are called
\emph{index coding}, has received particular attention~\cite{birk1, birk}.
The source at each time step is a vector of
independent and identically distributed (i.i.d.) bits,
each side information variable is a subset of these bits,
and the goal of each decoder is to losslessly reproduce
a subset of the bits that are not contained in its side information.
Treating the source as an i.i.d. vector of uniform bits is appropriate if the source is first compressed by an optimal rate-distortion encoder. Thus index coding implicitly assumes a separation-based architecture in which lossy compression is performed first and then the broadcasting  with side information is performed at the bit level. Ideally, one would consider both types of coding jointly. In previous work \cite{sinem_index}, we studied the index coding problem using tools from network 
information theory, in contrast to most past work on index coding which used techniques from network coding and graph theory. One of the advantages of using network 
information-theoretic tools, which was not pursued in the previous work, is that it allows one to consider the problems of lossy compression and coding for side information together, by allowing for a richer class of source models and distortion constraints. Our goal in this paper is to study systems that involve both lossy compression and coding for side information.

We shall focus on the case in which the source and the side
information at the decoders are all jointly Gaussian vectors.
This class of instances is important in applications, since
vector Gaussian sources are natural stepping stones on the
path from discrete memoryless sources to more sophisticated 
models of multimedia. The vector Gaussian setup can also 
be motivated theoretically since, like index coding, it is
one of the simplest classes of instances that are not degraded 
in general. We shall focus on the case of two decoders; unlike
index coding, for vector Gaussian problems even the two-decoder
case is nontrivial.

We provide a hierarchy of four lower bounds on the optimal
rate. For three separate special cases, we show that at least
one of the lower bounds matches the best-known achievable
rate~\cite{timo,berger}, thereby determining the optimal rate.
The four lower bounds are all obtained using variations on the
following argument. Since the rate-distortion function is known
when the side information is degraded~\cite{berger}, a natural approach to proving lower bounds is to \emph{enhance}
the side information of one encoder or the other in order to
make the problem degraded. The optimal rate for the newly-obtained instance is thus known and provides a lower bound on the optimal rate for the original instance. This idea can be applied several ways, leading to lower bounds of varying strength and usability. The weakest of these bounds is quite weak but also quite simple. The strongest, on the other hand, is quite strong but also difficult to apply. The intermediate bounds attempt to provide the best attributes of both.

We consider three different distortion constraints, all phrased
as constraints on the error covariance matrices, averaged over the block, at the two decoders. The first stipulates an
upper bound on the mean square error of the reproduction of
each component of the source; this can be viewed as constraints
on the diagonal elements of the time-average error covariance 
matrix. The second requires that the average error covariance
matrix itself must be dominated, in a positive definite sense,
by a given scaled identity matrix. In the final case, we 
require the trace of the average error covariance matrix
to be upper bounded by a constant. For each of the
three distortion measures, we solve a class of instances
using the lower bounds developed in the paper.
The necessary achievability arguments are standard, although
our analysis does provide insight into how the auxiliary
random variables therein should be chosen. Specifically,
we show how to divide the signal space into ``regions,''
in which the side information at one decoder is
``stronger'' than that of the other. We then show that
it is optimal for certain auxiliary random variables to
live in certain of these regions.  

The balance of the paper is organized as follows. 
The next two sections provide the problem formulation
and the four lower bounds,
respectively. Section \ref{sec:main_results} 
contains the statements of our optimality results for
all three cases described above. The achievability
analysis for these problems is presented in 
Section \ref{sec:ach_scheme}.  Section~\ref{sec:converse} 
shows how the lower bounds can be used to prove the
converse half of the optimality results.
Section~\ref{sec:conclusion} contains a
brief epilogue describing a conjectured difference 
among the lower bounds.

\section{Problem Definition}
\label{section:prob_defn}
Let $\mb{X}, \mb{Y_1}, \mb{Y_2}$ \footnote{We use bold letters to denote vectors.} be correlated vector Gaussian sources
\footnote{Unless otherwise is stated, we assume that all Gaussian random variables are zero mean.} of size $k\times1$, $k_1\times1$ and $k_2\times1$ respectively
 where  $\mb{X}$  is the source to be compressed at the encoder and  $\mb{Y_1}$
and $\mb{Y_2}$ comprise the side information at Decoder $1$ and Decoder $2$, respectively. We assume that the conditional covariance matrix of $\mb{X}$ given $\mb{Y_i}$,
 $K_{\mb{X}| \mb{Y_i}}$, $i \in \{1,2\}$ is invertible. Both Decoder $1$ and  $2$ wish to reconstruct $\mb{X}$ subject to given distortion constraints. 
The objective is to characterize the rate distortion function for this setting. The following definitions are used to  formulate the problem.

\begin{definition}
\label{defn:gamma}
$\Gamma_i$, $i \in \{1,2\}$ is defined as a mapping from the set of all $k \times k$ positive semi-definite (PSD) matrices to the set of $k_0 \times k_0$ PSD matrices such that
\\
1) $\Gamma_i(\cdot)$ is linear,
\\
2) $A \preceq B$\footnote{$A \preceq B$ means that $B-A$ is a positive semidefinite matrix.} implies that $\Gamma_i(A) \preceq \Gamma_i(B)$.
\end{definition}
\begin{definition}
An $(n, M, D_1, D_2)$ \emph{code} where $D_1$ and $D_2$ are positive
definite matrices,
is composed of 
\begin{itemize}
\item an \emph{encoding function}
\begin{align*}
f  : \mathbb{R}^{kn}  \rightarrow  \{ 1, ... , M \}
\end{align*}
\item and \emph{decoding functions}
\begin{align*}
g_1 & : \{ 1, ... , M \}  \times  \mathbb{R}^{k_1n}  \rightarrow \mathbb{R}^{kn} \\
g_2 & : \{ 1, ... , M \}  \times  \mathbb{R}^{k_2n}  \rightarrow \mathbb{R}^{kn} 
\end{align*}
\end{itemize}
satisfying the distortion constraints
\begin{align*}
 & E\left[\frac{1}{n} 
\sum_{k = 1}^{n} \Gamma_i\left((\mb{X}_k - \wh{\mb{X}}_{\mb{i}k})(\mb{X}_k - \wh{\mb{X}}_{\mb{i}k})^T\right)\right]  \preceq D_i  ,  \ \  i \in \{1,2\}
\end{align*}
where $\mb{\wh{X}^n_1} = g_1(f(\mb{X^n}), \mb{Y_1^n})$, and $\mb{\wh{X}^n_2} = g_2(f(\mb{X^n}), \mb{Y^n_2})$.
We call $n$ the \emph{block length} and $M$ the \emph{message size} of
the code.
\end{definition}
  
\begin{definition}
A rate $R$ is $(D_1, D_2)$-\emph{achievable} if for every $\epsilon > 0$, there exists an 
$(n,M, D_1+\epsilon I, D_2 + \epsilon I)$ \emph{code} such that $n^{-1}\log{M} \le R + \epsilon$ . 
\end{definition}

\begin{definition}
The \emph{rate-distortion function} is defined as
\begin{align*}
R(\mb{D}) = \inf\{R : R  \textrm{ is } \mb{D} \textrm{-achievable} \}, 
\end{align*}
where $\mb{D}=(D_1, D_2)$.
\end{definition}

We shall prove our lower bounds for arbitrary distortion measures $\Gamma$ satisfying the requirements of Definition~\ref{defn:gamma}. 
We conclude this section by introducing the following notations used in rest of the paper.
\begin{notation}
Let $\mb{X}$ be a $k\times 1$ vector where $k =l_1 +l_2$. Then
$(\mb{X})_{l_1}$ denotes the $l_1\times 1$ vector consisting of the first $l_1\times 1$ components of $\mb{X}$ and $[\mb{X}]_{l_2}$ denotes the remaining part of $\mb{X}$.
\end{notation}

\begin{notation}
Let $E$ be a $p\times p$ matrix. Then
$(E)_{ij}$ denotes the element of $E$ which is in the $i^{th}$ row and $j^{th}$ column of $E$.
\end{notation}

\begin{notation}
\label{not:l1}
Let $E$ and $F$ be $p\times p$ and $r\times r$ matrices where $p \geq l_1$ and $r \geq l_2$. Then
$(E)_{l_1}$ denotes the upper-left $l_1\times l_1$ submatrix of $E$ and 
 $[F]_{l_2}$ denotes the lower-right $l_2\times l_2$ submatrix of $F$. 
\end{notation}

\begin{notation}
\label{not:diag}
Let $E$ and $F$ be $p\times p$ and $r\times r$ matrices where $p \geq l_1$ and $r \geq l_2$. Then $(E)_{diag}$ denotes the $p\times p$ diagonal matrix whose diagonal elements are the same as  that of  $E$. Also,
$(E)_{l_1diag}$ denotes the $l_1\times l_1$ diagonal matrix whose diagonal elements are the same as  that of upper-left $l_1\times l_1$ submatrix of $E$ and 
 $[F]_{l_2diag}$ denotes the $l_2\times l_2$ diagonal matrix whose diagonal elements are the same as  that of lower-right $l_2\times l_2$ submatrix of $F$. 
\end{notation}

\begin{notation}
 Let $E$ and $F$ be $p\times p$ diagonal matrices. Then  
$\min\{E,F\}$ denotes the $p\times p$ diagonal matrix whose each diagonal entry is the minimum of corresponding diagonal entries of $E$ and $F$.
\end{notation}

\begin{notation}
Let $(\mb{X},\mb{\mb{Y}},\mb{Z})$ be a random vector. Then $\mb{X} \perp \mb{\mb{Y}} |\mb{Z}$ denotes that $\mb{X}$ and $\mb{\mb{Y}}$ are independent given $\mb{Z}$, $\mb{X} \lra \mb{\mb{Y}} \lra \mb{Z}$ denotes that $\mb{X}$, $\mb{\mb{Y}}$ and $\mb{Z}$ forms a Markov chain, and $K_{\mb{X}}$ denotes the covariance matrix of $\mb{X}$.
\end{notation}

\section{Lower Bounds}
\label{sec:lower}

We turn to lower bounds on the optimal rate. We shall provide 
four such bounds. In order of strongest (largest) to weakest (smallest),
these are
\begin{enumerate}
\item The \emph{Minimax bound} ($\mLB$);
\item The \emph{Maximin bound} ($\MLB$);
\item The \emph{Enhanced-Enhancement bound} (Enhanced-$\DSM$ or \ELB);
\item The \emph{Enhancement bound} ($\DSM$).
\end{enumerate}

Although the Maximin bound, the Enhanced-Enhancement bound, and the Enhancement bound
are never larger than the Minimax bound, they are useful in that
they are simpler to work with in some respects.  We begin
with the simplest, and weakest, of the bounds. This bound is 
folklore, and it turns out to be quite weak indeed.

\subsection{Enhancement Lower Bound}

If the side information at the decoders is \textit{degraded}, meaning
that we can find a joint distribution of $(\mb{X}, \mb{Y_1}, \mb{Y_2})$ 
such that
\begin{align}
\mb{X} \lra \mb{Y_{\sigma(1)}} \lra \mb{Y_{\sigma(2)}}
\end{align}
for some permutation $\sigma(.)$, then the rate distortion function is
known~\cite{berger, timo}.
Hence a natural way to obtain a lower bound to $R(\mb{D})$ is to 
create degraded problems by providing extra side information to one 
decoder or the other.
We call this lower bound \textit{enhancement lower bound}, abbreviated as $\DSM$, due to its similarity to the converse results for broadcast channels \cite{ vishwanath}. Proposition \ref{prop:DSM} states this lower bound.

\begin{proposition}
\label{prop:DSM}
The rate distortion function $R(\mb{D})$ is lower bounded by
 \begin{align}
 \label{eq:DSM}
R_{\DSM}(\mb{D}) = \max\{\sup_{S_G}\inf_{\wt{C}_{l1}(\mb{D})} R_{lo1}, \sup_{S_G} \inf_{\wt{C}_{l2}(\mb{D})} R_{lo2}\}, 
\end{align}
 where 
\begin{align}
\label{rlo1}
R_{lo1} &= I(\mb{X};\mb{W},\mb{U}|\mb{Y_1}) + I(\mb{X};\mb{V}|\mb{W},\mb{U},\mb{Y}), \\
\label{rlo2}
R_{lo2} &= I(\mb{X};\mb{W},\mb{V}|\mb{Y_2}) + I(\mb{X};\mb{U}|\mb{W},\mb{V},\mb{Y}),
\end{align} 
 $S_G = \{\mb{Y} \mbox{ jointly Gaussian with } (\mb{X},\mb{Y_1},\mb{Y_2}) | \mb{X} \lra \mb{Y} \lra (\mb{Y_1}, \mb{Y_2})   \}$,
 and
\begin{align*} 
&\wt{C}_{l1}(\mb{D})  \mbox{ is the set of } (\mb{W},\mb{U},\mb{V}) \mbox{ such that } &\mbox{ } & \wt{C}_{l2}(\mb{D}) 
\mbox{ is the set of } (\mb{W},\mb{U},\mb{V}) \mbox{ such that }
\\
& (\mb{W},\mb{U},\mb{V}) \lra \mb{X} \lra (\mb{Y},\mb{Y_1}, \mb{Y_2})  
&\mbox{ }
& (\mb{W},\mb{U},\mb{V}) \lra \mb{X} \lra (\mb{Y},\mb{Y_1}, \mb{Y_2}) 
 \\
& \Gamma_1\left(K_{\mb{X}|\mb{W},\mb{U},\mb{Y_1}}\right) \preceq D_1, \mbox{ } \Gamma_2\left(K_{\mb{X}|\mb{W},\mb{U},\mb{V},\mb{Y}} \right)\preceq D_2  
&\mbox{ }
& \Gamma_1\left(K_{\mb{X}|\mb{W},\mb{U},\mb{V} ,\mb{Y}}\right) \preceq D_1, \mbox{ }  \Gamma_2\left(K_{\mb{X}|\mb{W},\mb{V},\mb{Y_2}}\right) \preceq D_2. 
\end{align*}
\end{proposition}

The $\DSM$ is quite weak. Consider, for example, what is arguably
the simplest nontrivial instance of the problem: the source $\mb{X}$ is bivariate,  $K_{\mb{X}}$, $K_{\mb{X}|\mb{Y_1}}$, 
and $K_{\mb{X}|\mb{Y}_2}$ are all diagonal, and the reconstructions at decoders are subject to component-wise MSE distortion constraints.
This is essentially the
parallel scalar Gaussian version of the problem. If the overall
problem is degraded then the $\DSM$ is of course tight. But if one
of the two components is degraded in one direction and the other
component is degraded in the other, then Watanabe~\cite{watanabe}
has shown that the $\DSM$ is not tight, at least for the natural
choice of $\mb{Y}$ that has
$$
K_{\mb{X}|\mb{Y}} = \min\left(K_{\mb{X}|\mb{Y_1}}, K_{\mb{X}|\mb{Y}_2}\right).
$$

Comparing the $\DSM$ against the achievable bound in 
Theorem~\ref{thm:achievable} to follow, one sees
several potential sources of looseness. We shall see that the culprit is that the distortion constraints
\begin{align*}
\Gamma_1\left(K_{\mb{X}|\mb{W},\mb{U},\mb{Y_1}}\right) & \preceq D_1
\\
\Gamma_2\left(K_{\mb{X}|\mb{W},\mb{V},\mb{Y_2}} \right) & \preceq D_2  
\end{align*}
in the achievable bound in Theorem  \ref{thm:achievable} have been weakened to
\begin{align*}
\Gamma_1\left(K_{\mb{X}|\mb{W},\mb{U},\mb{V} ,\mb{Y}}\right) 
& \preceq D_1
\\
\Gamma_2\left(K_{\mb{X}|\mb{W},\mb{U},\mb{V},\mb{Y}} \right) 
   & \preceq D_2  
\end{align*}
here.
Weakening the constraints in this way allows less informative
$(\mb{W},\mb{U},\mb{V})$ to be feasible, because one can
make use of the enhanced side information $\mb{Y}$ for
estimation purposes. We shall make this intuition precise by 
showing that the Maximin and Enhanced-Enhancement lower bound, which differ from the $\DSM$  only in the distortion constraints, are tight for this problem. For reasons of expeditiousness, we shall state and prove the Minimax lower bound first, and then weaken it to obtain the Maximin and  Enhanced-Enhancement lower bound.

\subsection{Minimax Lower Bound}
\label{subsec:minimax}

Theorem \ref{thm:minmax} states the \textit{Minimax lower bound}, abbreviated as $\mLB$, to the rate distortion problem.
\begin{theorem}
\label{thm:minmax}
 The rate distortion function, $R(\mb{D})$, is lower bounded by
 \begin{align}
 \label{eq:mLB}
R_{\mLB}(\mb{D}) = \sup_{S } \inf_{C_l(\mb{D})} \max \{R_{lo1}, R_{lo2} \} 
\end{align}
 where $R_{lo1}$ and $R_{lo2}$ are as in (\ref{rlo1}) and (\ref{rlo2}), and 
\begin{align*}
&S = \{\mb{Y} | \mb{X} \lra \mb{Y} \lra(\mb{Y_1}, \mb{Y_2})\}
\\
&C_l(\mb{D}) \mbox{: the set of } (\mb{W},\mb{U},\mb{V}) \mbox{ such that }
\\
& (\mb{W},\mb{U},\mb{V}) \lra \mb{X} \lra (\mb{Y_1}, \mb{Y_2}, \mb{Y})  
 \\
&\Gamma_1\left( K_{\mb{X}|\mb{W},\mb{U},\mb{Y_1}}\right)  \preceq D_1, \Gamma_2\left(K_{\mb{X}|\mb{W},\mb{V},\mb{Y_2}} \right) \preceq D_2.
\end{align*}
\end{theorem}
\begin{proof}[Proof of Theorem \ref{thm:minmax} ]
By definition, for any $\mb{D}-$achievable rate, $R$, and for all $\epsilon > 0$, we can find a   $(n,2^{n(R + \epsilon)}, \mb{D} + \epsilon (I,I))$ code. Let $\epsilon > 0$ be given and $J$ denote the output of the encoder. Also let $\mb{Y}$ be an auxiliary source in  $S$.  Then, we can write
\begin{align}
n(R + \epsilon) &\ge H(J)  
\notag \\
&\ge I(\mb{X}^n, \mb{Y}^n_\mb{1}, \mb{Y}^n; J) 
 \notag \\
&\overset{a}{=} I(\mb{Y}^n_\mb{1}; J) + I(\mb{Y}^n; J| \mb{Y}^n_\mb{1}) + I(\mb{X}^n; J| \mb{Y}^n_1,\mb{Y}^n) 
 \notag \\
&\ge I(\mb{Y}^n; J| \mb{Y}^n_\mb{1}) + I(\mb{X}^n; J| \mb{Y}^n_\mb{1},\mb{Y}^n) 
\notag \\
\label{ineq_1}
&\overset{b}{\ge} \sum^{n}_{i=1} \big{[} I( \mb{Y}_{i}; J,\mb{Y}_{\mb{1}\hc{i}}|\mb{Y}_{\mb{1}i})+  I( \mb{X}_{i}; J,\mb{Y}_{\mb{1}\hc{i}},\mb{Y}_{\hc{i}}|\mb{Y}_{\mb{1}i},\mb{Y}_i) \big{]}
\end{align}
where $\mb{Y}_{\mb{1}\hc{i}}$ denotes all $\mb{Y}^n_\mb{1}$ except $\mb{Y}_{1i}$ and $a$  is due to the chain rule, and $b$ is due to the chain rule and that conditioning reduces entropy.
Then if we apply the chain rule to the last term above, the right hand side of (\ref{ineq_1}) equals
\begin{align}
& \sum^{n}_{i=1} \big{[} I( \mb{Y}_{i}; J,\mb{Y}_{\mb{1}\hc{i}}|\mb{Y}_{\mb{1}i})+  I( \mb{X}_{i}; J,\mb{Y}_{\mb{1}\hc{i}}|\mb{Y}_{\mb{1}i},\mb{Y}_i) 
+  I( \mb{X}_{i}; \mb{Y}_{\hc{i}}|J,\mb{Y}_{\mb{1}\hc{i}}, \mb{Y}_{\mb{1}i},\mb{Y}_i) \big{]}
\\ \notag
%
& = \sum^{n}_{i=1} \big{[} I( \mb{X}_i,\mb{Y}_{i}; J,\mb{Y}_{\mb{1}\hc{i}}|\mb{Y}_{\mb{1}i})
+  I( \mb{X}_{i}; \mb{Y}_{\hc{i}}|J,\mb{Y}_{\mb{1}\hc{i}}, \mb{Y}_{\mb{1}i},\mb{Y}_i) \big{]}
\notag \\
\label{lb:ineq_2}
&\ge  \sum^{n}_{i=1} \big{[} I( \mb{X}_i; J,\mb{Y}_{\mb{1}\hc{i}}|\mb{Y}_{1i})
+  I( \mb{X}_{i}; \mb{Y}_{\hc{i}}|J,\mb{Y}_{\mb{1}\hc{i}}, \mb{Y}_{\mb{1}i},\mb{Y}_i) \big{]}.
\end{align}

Also, since $\mb{X} \lra \mb{Y} \lra (\mb{Y_1},\mb{Y_2})$ the right hand side of (\ref{lb:ineq_2}) is equal to
\begin{align}
& \sum^{n}_{i=1} \big{[} I( \mb{X}_i; J,\mb{Y}_{\mb{1}\hc{i}}|\mb{Y}_{\mb{1}i})
+  I( \mb{X}_{i}; \mb{Y}_{\hc{i}}|J,\mb{Y}_{\mb{1}\hc{i}},\mb{Y}_i) \big{]}
\notag \\
&{\ge}\sum^{n}_{i=1} \big{[} I( \mb{X}_i; J,\mb{Y}_{\mb{1}\hc{i}}|\mb{Y}_{\mb{1}i})
+  I( \mb{X}_{i}; \mb{Y}_{\mb{2}{\hc{i}}}|J,\mb{Y}_{\mb{1}\hc{i}},\mb{Y}_i) \big{]}
\notag \\
\label{lb:ineq_3}
&{=} \sum^{n}_{i=1} \big{[} I( \mb{X}_i; \mb{W'}_i,\mb{U'}_i|\mb{Y}_{\mb{1}i})
+  I( \mb{X}_{i}; \mb{V'}_i|\mb{W'}_i,\mb{U'}_i,\mb{Y}_i) \big{]}
\end{align}
where $\mb{W'}_i = J$,  $\mb{U'}_i = \mb{Y}_{\mb{1}\hc{i}}$ and $\mb{V'}_i = \mb{Y}_{\mb{2}\hc{i}}$.
Note that $(\mb{W'}_i, \mb{U'}_i, \mb{V'}_i) \lra \mb{X}_i \lra (\mb{Y}_{\mb{1}i}, \mb{Y}_{\mb{2}i}, \mb{Y}_{i})$ for all $i \in [n]$. 
Let $T$ be a random variable uniformly distributed on $[n]$ and independent of the source, side information and  all $(\mb{W}'_i, \mb{U}'_i,\mb{V}'_i)$, $i \in [n]$. Then we can write the right hand side of (\ref{lb:ineq_3}) as
\begin{align}
& \sum^{n}_{i=1} \big{[} I( \mb{X}_i; \mb{W'}_i,\mb{U'}_i|\mb{Y}_{\mb{1}i}, T=i)
+  I( \mb{X}_{i}; \mb{V'}_i|\mb{W'}_i,\mb{U'}_i,\mb{Y}_i, T=i) \big{]}
\notag \\
&= n \big{[} I( \mb{X}; \mb{W'},\mb{U'}, T|\mb{Y}_\mb{1})
+  I( \mb{X}; \mb{V'}, T|\mb{W'},\mb{U'},T,\mb{Y}) \big{]}
\notag \\
&= nR_{lo1}, \mbox{ by denoting $(\mb{W'},T)$, $(\mb{U'},T)$, $(\mb{V'},T)$ as $\mb{W}$, $\mb{U}$, $\mb{V}$ respectively.}
\end{align}

If we swap the role of $\mb{Y_1}$ and $\mb{Y_2}$ and apply the same procedure above, we can get

\begin{align}
R + \epsilon &\ge I( \mb{X}; \mb{W},\mb{V}|\mb{Y}_{2})
+  I( \mb{X}; \mb{U}|\mb{W},\mb{V},\mb{Y}) \notag \\
&= R_{lo2}.
\end{align}

Note that since $(\mb{W'}_i, \mb{U'}_i, \mb{V'}_i) \lra \mb{X}_i \lra (\mb{Y}_{\mb{1}i}, \mb{Y}_{\mb{2}i}, \mb{Y}_{i})$ for all $i \in [n]$, we have $(\mb{W}, \mb{U}, \mb{V}) \lra \mb{X} \lra (\mb{Y_1}, \mb{Y_2}, \mb{Y})$.  Moreover since $(\mb{W'}_i,\mb{U'}_i, \mb{Y}_{\mb{1}i}) = (J, \mb{Y}^n_\mb{1})$ and $(\mb{W'}_i,\mb{V'}_i, \mb{Y}_{\mb{2
}i}) = (J, \mb{Y}^n_\mb{2})$,  given $(\mb{W'}_i,\mb{U'}_i, \mb{Y}_{\mb{1}i})$
Decoder $1$ can reconstruct the source, $\mb{X}_i$, subject to its distortion constraint. Similarly,  Decoder $2$ can reconstruct the source, $\mb{X}_i$ given $(\mb{W'}_i,\mb{V'}_i, \mb{Y}_{\mb{2}i})$.
Hence, $(\mb{W},\mb{U},\mb{V}) \in C_{l}(\mb{D} +  \epsilon(I,I))$ and we have 
\begin{align}
\label{ineq:minmaxlb_epsilon}
&R(\mb{D}) \ge \inf_{C_l(\mb{D} +  \epsilon(I,I))} \max \{R_{lo1}, R_{lo2} \} - \epsilon.
\end{align}

Let $R'_{lo}(\mb{D} + \epsilon(I,I), \mb{Y})$ denote the right hand side of (\ref{ineq:minmaxlb_epsilon}). Note that (\ref{ineq:minmaxlb_epsilon}) holds for any $\mb{Y} \in S$,  where $S$ as in Theorem \ref{thm:minmax}.
Hence we can write
\begin{align}
\label{ineq:minmaxlb_wepsilon}
&R(\mb{D}) \ge \sup_{S} R'_{lo}(\mb{D} + \epsilon(I,I), \mb{Y})  - \epsilon.
\end{align}

Note that from  Lemma \ref{lemma:convexity} in Appendix \ref{app:convexity}, $R'_{lo}(\mb{D},  \mb{Y})$ is  convex in $\mb{D}$. Since $0 \prec D_i$, $i \in \{1,2\}$ we can find $\delta(D_1,D_2) > 0$   such that $0 \prec D_i - \delta(D_1,D_2) I$ for $i \in \{1,2\}$. Hence $R'_{lo}(\mb{D} + \gamma(I,I),  \mb{Y})$ is also convex in $\gamma$, where $\gamma \ge -\delta(D_1,D_2)$. Note that  $\sup_{S} R'_{lo}(\mb{D} +\epsilon(I,I), \mb{Y})$ is also convex since supremum of convex functions is convex. Then, we can conclude that $\sup_{S} R'_{lo}(\mb{D} +\epsilon(I,I), \mb{Y})$ is continuous at $\epsilon = 0$ since a convex function on an open set is continuous.  Lastly, since $\epsilon$ was arbitrary, letting $\epsilon \rightarrow 0$ gives the result. 
\end{proof}

It is worth noting that one can prove a  bound similar to $\mLB$ for non-Gaussian 
sources and general additive distortion constraints.
Although the $\mLB$ is quite powerful, it can
be difficult to apply. In particular, it is not clear that
it is sufficient to consider $(\mb{W},\mb{U},\mb{V})$
that are jointly Gaussian with $(\mb{X},\mb{Y}_1,\mb{Y}_2)$.
Similarly, when considering the analogous form of this bound
for discrete memoryless sources, it it not clear how to
obtain cardinality bounds on the auxiliary random variables
$(\mb{W},\mb{U},\mb{V})$. As such, it is not clear how to
compute this bound in general. We shall therefore consider
a slightly weakened form of the bound that is easier to apply.
It turns out that simply swapping the $\min$ and the $\max$
in the objective and adding that $\mb{Y}$ is jointly Gaussian with $(\mb{X}, \mb{Y_1},\mb{Y_2})$ to $S$ yields a bound that is significantly more tractable.

\subsection{Maximin Lower Bound}
\label{subsec:maximin}
The next proposition gives the \textit{Maximin lower bound}, abbreviated as $\MLB$. 
\begin{proposition}
\label{prop:maximin}
The rate distortion function, $R(\mb{D})$, is lower bounded by
 \begin{align}
 \label{eq:MLB}
R_{MLB}(\mb{D}) = \max\{ \sup_{S_G}\inf_{C_{l1}(\mb{D})} R_{lo1}, \sup_{S_G}\inf_{C_{l2}(\mb{D})} R_{lo2}\}, 
\end{align}
 where $R_{lo1}$ and $R_{lo2}$ are as in (\ref{rlo1}) and (\ref{rlo2}) respectively,
 $S_G$ as in Proposition \ref{prop:DSM}, and
\begin{align*} 
&C_{l1}(\mb{D}) : \mbox{ the set of } (\mb{W},\mb{U},\mb{V}) \mbox{ such that } &\mbox{ } & C_{l2}(\mb{D}) :  \mbox{ the set of } (\mb{W},\mb{U},\mb{V}) \mbox{ such that }
\\
& (\mb{W},\mb{U},\mb{V}) \lra \mb{X} \lra (\mb{Y_1}, \mb{Y_2},\mb{Y})  
&\mbox{ }
& (\mb{W},\mb{U},\mb{V}) \lra \mb{X} \lra (\mb{Y_1}, \mb{Y_2}, \mb{Y}) 
 \\
& \Gamma_1\left(K_{\mb{X}|\mb{W},\mb{U},\mb{Y_1}}\right) \preceq D_1, \mbox{ } \Gamma_2\left(K_{\mb{X}|\mb{W},\mb{V},\mb{Y_2}} \right)\preceq D_2  
&\mbox{ }
& \Gamma_1\left(K_{\mb{X}|\mb{W},\mb{U}, \mb{Y_1}}\right) \preceq D_1, \mbox{ }  \Gamma_2\left(K_{\mb{X}|\mb{W},\mb{V},\mb{Y_2}}\right) \preceq D_2. 
\end{align*}
\end{proposition}
\begin{proof}
This follows directly from the $\mLB$, Theorem~\ref{thm:minmax}, by
moving the $\inf$ in the objective inside the maximization over the bounds in (\ref{rlo1}) and (\ref{rlo2}) and replacing the set $S$ with $S_G$.
\end{proof}

Although numerical evidence suggests that the $\MLB$ can
be strictly weaker than the $\mLB$ (see the discussion
in Section~\ref{sec:conclusion} to follow), the $\MLB$ does have
certain advantages. For the analogous bound for discrete memoryless
sources with additive distortion measures, one can obtain 
cardinality bounds on the alphabets of $\mb{W}$, $\mb{U}$, 
and $\mb{V}$ using straightforward techniques~\cite{Csiszar}. And we shall
show that, for the Gaussian form examined here, one may
restrict attention to $\mb{W}$, $\mb{U}$, and $\mb{V}$ that
are jointly Gaussian with $(\mb{X}, \mb{Y}_1, \mb{Y}_2, \mb{Y})$.

Evidently the $\MLB$ differs from the $\DSM$ in Proposition \ref{prop:DSM} only in that the distortion constraints 
are replaced with those that appear in the achievable upper
bound presented in Theorem \ref{thm:achievable} in Section \ref{sec:ach_scheme}. In Section \ref{sec:converse}, we shall see that this improvement suffices
to make the bound tight for the rate distortion problem with MSE distortion constraints stated in Section \ref{sec:main_results}. We turn to the
fourth and final lower bound.

\subsection{Enhanced-Enhancement Lower Bound}
\label{subsec:ELB}
\begin{proposition}
\label{prop:ELB}
The rate distortion function, $R(\mb{D})$, is lower bounded by
 \begin{align}
 \label{eq:ELB}
R_{\ELB}(\mb{D}) = \max\{ \sup_{S_G}\inf_{\bar{C}_{l1}(\mb{D})} R_{lo1}, \sup_{S_G}\inf_{\bar{C}_{l2}(\mb{D})} R_{lo2}\}, 
\end{align}
  where $R_{lo1}$ and $R_{lo2}$ are as in (\ref{rlo1}) and (\ref{rlo2}) respectively, $S_G$ as in Proposition \ref{prop:DSM}, and
\begin{align*} 
&\bar{C}_{l1}(\mb{D}) : \mbox{ the set of } (\mb{W},\mb{U},\mb{V}) \mbox{ such that }&\mbox{ } & \bar{C}_{l2}(\mb{D}) :  \mbox{the set of } (\mb{W},\mb{U},\mb{V}) \mbox{ such that }
\\
& (\mb{W},\mb{U},\mb{V}) \lra \mb{X} \lra (\mb{Y_1}, \mb{Y_2},\mb{Y})  
&\mbox{ }
& (\mb{W},\mb{U},\mb{V}) \lra \mb{X} \lra (\mb{Y_1},\mb{Y_2}, \mb{Y}) 
 \\
& \Gamma_1\left(K_{\mb{X}|\mb{W},\mb{U},\mb{Y_1}}\right) \preceq D_1, 
&\mbox{ }
&\Gamma_1\left(( K^{-1}_{\mb{X}|\mb{W},\mb{U},\mb{V}, \mb{Y}} - \wh{K})^{-1}\right)\preceq D_1,
\\
& \Gamma_2\left((K^{-1}_{\mb{X}|\mb{W},\mb{U},\mb{V},\mb{Y}} - \wt{K})^{-1} \right)\preceq D_2 
&\mbox{ }
& \Gamma_2\left( K_{\mb{X}|\mb{W},\mb{V},\mb{Y_2}} \right)\preceq D_2
\end{align*}
and $\wt{K} = K^{-1}_{\mb{X} | \mb{Y}} - K^{-1}_{\mb{X} | \mb{Y_2}}$,   $\wh{K} = K^{-1}_{\mb{X} | \mb{Y}} - K^{-1}_{\mb{X} | \mb{Y_1}}$.
\end{proposition}
\begin{proof}
Note that only difference between $\MLB$ and Enhanced-$\DSM$ is  the optimization sets over which the infima are taken. Hence it is enough to show that $C_{li}(\mb{D}) \subseteq \bar{C}_{li}(\mb{D})$ for $i \in \{1,2\}$. Let $(\mb{W},\mb{U}, \mb{V}) \in C_{l1}(\mb{D})$. 
Then $(\mb{W},\mb{U}, \mb{V})$ satisfy the Markov chain condition $(\mb{W},\mb{U},\mb{V}) \lra \mb{X} \lra (\mb{Y_1}, \mb{Y_2},\mb{Y})$ and we have  $\Gamma_1\left(K_{\mb{X}|\mb{W},\mb{U},\mb{Y_1}}\right) \preceq D_1$. 
Also, the inequalities $K_{\mb{X}|\mb{W},\mb{U},\mb{V},\mb{Y_2}} \preceq K_{\mb{X}|\mb{W},\mb{V},\mb{Y_2}}$ and $K^{-1}_{\mb{X}|\mb{W},\mb{U},\mb{V},\mb{Y_2}} \preceq K^{-1}_{\mb{X}|\mb{W},\mb{U},\mb{V},\mb{Y}} - \wt{K}$ imply, by the Gaussian ``variance-drop'' lemma (Lemma~\ref{lemma:inequality} in Appendix~\ref{app:ELB}), that $\Gamma_2\left((K^{-1}_{\mb{X}|\mb{W},\mb{U},\mb{V},\mb{Y}} - \wt{K})^{-1} \right)\preceq D_2 $. Hence  $(\mb{W},\mb{U}, \mb{V})$ is also in $\bar{C}_{l1}(\mb{D})$, giving $C_{l1}(\mb{D}) \subseteq \bar{C}_{l1}(\mb{D})$. We can apply similar procedure to get $C_{l2}(\mb{D}) \subseteq \bar{C}_{l2}(\mb{D})$, which concludes the proof.
\end{proof}

Comparing the Enhanced-$\DSM$ against the $\DSM$ in (\ref{eq:DSM}) shows that the
differences lie entirely in the distortion constraints. The $\DSM$ effectively allows the decoders to use their ``enhanced''
side information for the purposes of estimating the source.
The achievable bound, by contrast, does not.
The Enhanced-$\DSM$ allows the decoders
to use their enhanced side information, but it also tightens the
constraint to account for this extra information, as justified
by the Gaussian variance-drop lemma. We shall see in the next subsection
that the Enhanced-$\DSM$ actually coincides with the 
$\MLB$ for all of the problems considered in this
paper. We mention the Enhanced-$\DSM$ only because the idea
of using the Gaussian variance-drop lemma to tighten the distortion constraints at decoders that are provided with improved side information may prove useful in other contexts.

\subsection{Properties of the Lower Bounds}
\label{subsec:comparison}

It is evident from the proofs in this section
that the four lower bounds can be ordered as follows
$$
R_{\DSM}(\mb{D}) \le
R_{\ELB}(\mb{D}) \le
R_{\MLB}(\mb{D}) \le
R_{\mLB}(\mb{D}).
$$

We shall show that Gaussian auxiliary random variables are
optimal for $\MLB$, Enhanced-$\DSM$, and $\DSM$, and that the $\MLB$ and 
Enhanced-$\DSM$ are in
fact equal. We begin by showing that Gaussian auxiliary
random variables are optimal for the $\DSM$ and Enhanced-$\DSM$.
\begin{lemma}
\label{lemma:gauss_ELB}
One may add the constraint that $(\mb{W},\mb{U},\mb{V})$ is jointly
Gaussian with $(\mb{X}, \mb{Y}_1, \mb{Y}_2, \mb{Y})$ to the optimization
problem in the  $\DSM$ in (\ref{eq:DSM}) and the Enhanced-$\DSM$ 
in (\ref{eq:ELB}) without
affecting the optimal value.
\end{lemma}
\begin{proof}
See Appendix~\ref{app:ELB_bound}.
\end{proof}

\begin{proposition}
\label{prop:ELB=maxmin}
The Maximin bound and Enhanced-Enhancement bound in Proposition \ref{prop:maximin} and \ref{prop:ELB}, respectively, coincide:
 \begin{equation}
 R_{\MLB}(\mb{D}) = R_{\ELB}(\mb{D}). 
 \end{equation}
\end{proposition}

\begin{proof}
It suffices to show that 
$$
R_{\MLB}(\mb{D}) \le R_{\ELB}(\mb{D})
$$
By Lemma \ref{lemma:gauss_ELB}, $(\mb{W},\mb{U},\mb{V})$ in  $\bar{C}_{l1}(\mb{D})$ or $\bar{C}_{l2}(\mb{D})$ can be restricted to vector Gaussian random variables without loss of optimality. Furthermore, any $\mb{U} \in  \bar{C}_{l1}$  can be lumped into $\mb{W} \in  \bar{C}_{l1}(\mb{D})$, i.e. $\mb{U}$ is deterministic, without loss of optimality since $\mb{W}$ and $\mb{U}$ always appear together both in the objective and the conditions. The same argument holds when we swap the roles of $\mb{U}$  and $\mb{V}$ in $\bar{C}_{l2}(\mb{D})$.
Hence, with those additional conditions we can write the optimizing sets, $\bar{C}_{l1}(\mb{D})$ and $\bar{C}_{l2}(\mb{D})$, as
\begin{align*}
&\bar{C}_{l1}(\mb{D}) : &\mbox{ } & \bar{C}_{l1}(\mb{D}) :  
\\
& (\mb{W},\mb{U},\mb{V}) \lra \mb{X} \lra (\mb{Y_1}, \mb{Y_2},\mb{Y})  
&\mbox{ }
& (\mb{W},\mb{U},\mb{V}) \lra \mb{X} \lra (\mb{Y_1}, \mb{Y_2},\mb{Y}) 
 \\
 & (\mb{W},\mb{U},\mb{V},\mb{X},\mb{Y_1}, \mb{Y_2},\mb{Y}) \mbox{ jointly Gaussian },   \mb{U} = \emptyset
&\mbox{ }
& (\mb{W},\mb{U},\mb{V},\mb{X},\mb{Y_1},\mb{Y_2}, \mb{Y}) \mbox{ jointly Gaussian },  \mb{V} = \emptyset
 \\
& K_{\mb{X}|\mb{W},\mb{Y_1}} \preceq D_1, \mbox{ } K_{\mb{X}|\mb{W}, \mb{V},\mb{Y_2}} \preceq D_2  
&\mbox{ }
& K_{\mb{X}|\mb{W},\mb{U}, \mb{Y_1}} \preceq D_1, \mbox{ }  K_{\mb{X}|\mb{W},\mb{Y_2}} \preceq D_2
\end{align*}

  Then any such $(\mb{W},\mb{U},\mb{V}) \in \bar{C}_{l1}(\mb{D})$ (or $(\mb{W},\mb{U},\mb{V}) \in \bar{C}_{l2}(\mb{D})$)  is also in ${C}_{l1}(\mb{D})$ (or ${C}_{l2}(\mb{D})$).  Hence, $R_{\MLB}(\mb{D}) \le  R_{\ELB}(\mb{D})$.
\end{proof}

It follows from the two previous results that Gaussian auxiliary
random variables are optimal for the $\MLB$. To see this,
let $R^G_{\ELB}(\mb{D})$ denote the Enhanced-$\DSM$ with
the auxiliary random variables constrained to be jointly Gaussian
with $(\mb{X},\mb{Y}_1,\mb{Y}_2)$. Define $R^G_{\MLB}(\mb{D})$ 
likewise. Then we have
\begin{align*}
R^G_{\MLB}(\mb{D})
& \ge R_{\MLB}(\mb{D}) \\
& \overset{a}{=} R_{\ELB}(\mb{D}) \\
& \overset{b}{=} R^G_{\ELB}(\mb{D}) \\
& \overset{c}{=} R^G_{\MLB}(\mb{D}),
\end{align*}
where $a$ follows from Proposition~\ref{prop:ELB=maxmin}, $b$ follows from
Lemma~\ref{lemma:gauss_ELB}, and $c$ is straightforward to verify.

We now proceed to state our optimality results.
\section{Optimality Results}
\label{sec:main_results}
We shall determine the optimal rate for the following choices
of $\Gamma_1$, $\Gamma_2$, $D_1$, and $D_2$:
\begin{enumerate}
	\item \emph{Mean square error (MSE):}  $\Gamma_1$ and $	\Gamma_2$ are chosen as
	\begin{align}
	\Gamma_i(K) = (K)_{diag} \quad i \in \{1,2\},
	\label{dm_MSE}
	\end{align}
 and $D_1$ and $D_2$ are diagonal matrices satisfying
\begin{equation}
	\label{MSE_distcons}
	D_1 \preceq K_{\mb{X}|\mb{Y_1}} \mbox{ and } \  D_2 \preceq K_{\mb{X}|	\mb{Y_2}}.
\end{equation}

\item \emph{Error covariance matrix:} $\Gamma_1$ and $\Gamma_2$ are chosen as
	\begin{align}
	\label{dm_SI}
	\Gamma_i(K) & = K \quad i \in \{1,2\}
	\end{align}
and $D_1$ and $D_2$ are scaled identity matrices satisfying 
\begin{equation}
	\label{si_distcons}
	D_1 \preceq K_{\mb{X}|\mb{Y_1}} \mbox{ and } \  D_2 \preceq K_{\mb{X}|	\mb{Y_2}}.
\end{equation}

Note that scaled identity matrix constraints on the error covariance matrix enable us to bound the MSE of the reconstruction vector uniformly from all directions.
\item \emph{Trace of the error covariance matrix:} $\Gamma_1$ and $\Gamma_2$ are chosen as
\begin{align}
\label{dm_trace}
\Gamma_i(K) = \Tr(K) \quad i \in \{1,2\},
\end{align}
and $D_1$ and $D_2$ are scalars satisfying
\begin{align}
\label{tr_distcons}
D_1I \preceq K_{\mb{X}|\mb{Y_1}} \mbox{ and } D_2I \preceq K_{\mb{X}|\mb{Y_2}}.
\end{align} 
\end{enumerate}

Most of the prior work on the Heegard-Berger problem assumes some sort of degradedness structure between the source and the side information at the two decoders (e.g. \cite{berger, watanabe, timo_lessnoisy}). Watanabe~\cite{watanabe}, in particular, assumes that the source and the side information all consist of two components, and the first components of all three variables are independent of the second components of all three variables. The two components are ``mismatched degraded,'' i.e., each component is individually degraded, but the two components are degraded in opposite order. Although we do not assume any degradedness structure, we shall reduce our problems to one that resembles Watanabe's. Specifically, we shall decompose the signal space into  ``\textit{regions}," one of which is such that the side information at Decoder 1 is ``\textit{stronger}" than that of Decoder 2 and one such that the reverse is true. Many such candidate decompositions are possible;
we shall use the following one.

Recall that we assume that $K_{\mb{X}|\mb{Y_i}}$, $i \in \{1,2\}$ are invertible matrices.\footnote{The distortion constraints in (\ref{si_distcons}),  (\ref{tr_distcons}), and (\ref{MSE_distcons}) also imply that $K_{\mb{X}|\mb{Y_1}}$ and $K_{\mb{X}|\mb{Y_2}}$ are positive definite matrices. } Now consider the matrix $K^{-1}_{\mb{X}|\mb{Y_2}} - K^{-1}_{\mb{X}|\mb{Y_1}}$. Since it is symmetric we can find an orthogonal matrix $Q_1$ such that $Q_1(K^{-1}_{\mb{X}|\mb{Y_2}} - K^{-1}_{\mb{X}|\mb{Y_1}})Q_1^T$
is diagonal.  Furthermore, we can find another orthogonal matrix $Q_2$ such that $Q_2Q_1(K^{-1}_{\mb{X}|\mb{Y_2}} - K^{-1}_{\mb{X}|\mb{Y_1}})Q_1^TQ_2^T$ is of the form
\begin{align}
	\label{relation}
 	&K=  \left( \begin{array}{c c} A & \mf{0} \\ \mf{0} & B 	\end{array} \right)
\end{align}
 where $A \succeq 0$ is an $l_1\times l_1$ diagonal matrix, $ B \prec 0$ is an $l_2\times l_2$ diagonal matrix  and $l_1+l_2 =k$.
 
Let $Q= Q_2Q_1$.  Note that $QDQ^T= D$ when $D$ is a scaled identity matrix  and  distortion measure in (\ref{dm_trace}) is invariant under $(\mb{X}, \mb{\wh{X}_i}) \rightarrow (Q\mb{X}, Q\mb{\wh{X}_i})$.

Note that MSE distortion measure is not invariant under $(\mb{X}, \mb{\wh{X}_i}) \rightarrow (Q\mb{X}, Q\mb{\wh{X}_i})$. Then for MSE and any $\Gamma_i$ such that it is not invariant under $(\mb{X}, \mb{\wh{X}_i}) \rightarrow (Q\mb{X}, Q\mb{\wh{X}_i})$, we restrict our attention to the source $\mb{X}$ and side information $\mb{Y_i}$ such that 
 \begin{align}
	\label{eq:relation}
 K^{-1}_{\mb{X}|\mb{Y_2}} - K^{-1}_{\mb{X}|\mb{Y_1}} = K.
 \end{align}
 
Therefore, the rate-distortion problems where $Q\mb{X}$ is the source, $\mb{Y_i}$ is side information at Decoder $i$ subject to the distortion constraints $D_i$, $i \in \{1,2\}$  are equivalent to the problems that we defined at the beginning. For the rest of the paper, we assume that $Q\mb{X}$ is the source and we relabel $Q\mb{X}$ as $\mb{X}$ for the ease of notation, $\mb{Y_1}$ and $\mb{Y_2}$ are side information and $D_1$  and $D_2$ distortion constraints for Decoder $1$ and $2$, respectively, as shown in Figure \ref{fig:setup}. Note that we have not entirely reduced the problem to that of Watanabe because the components of $\mb{X}$ may be dependent.

From now on we use the abbreviation $\RDSI$ for the problem of finding the rate distortion function where reconstructions at decoders are subjected to error covariance distortion constraints that are scaled identity matrices as in (\ref{si_distcons}) and denote  the corresponding rate distortion function as $R^{\si}(\mb{D})$, where $\mb{D} = (D_1,D_2)$. Also $\RDTr$ and $\RDmse$ denote the  rate distortion problems where decoders have distortion constraints as in (\ref{tr_distcons}) on the  trace of error covariance matrices and (\ref{MSE_distcons}) componentwise MSE constraints, respectively. The corresponding rate distortion functions for $\RDTr$ and $\RDmse$ are denoted by $R^{\tr}(\mb{D})$ and $R^{\mse}(\mb{D})$, respectively.

\begin{remark}
%
Since $(K^{-1}_{\mb{X}| \mb{Y_2}})_{l_1} \succeq (K^{-1}_{\mb{X}| \mb{Y_1}})_{l_1} $, we say that $\mb{Y_2}$ is ``\textit{stronger}" than $\mb{Y_1}$ in the ``\textit{region}" involving the upper-left part of the inverse covariance matrices. Similarly, $\mb{Y_1}$ is ``\textit{stronger}" than $\mb{Y_2}$ in the lower-right part of the inverse covariance matrices since  $[K^{-1}_{\mb{X}| \mb{Y_2}}]_{l_2} \preceq [K^{-1}_{\mb{X}| \mb{Y_1}}]_{l_2}$.
\end{remark}
Now we are ready to state our optimality results.
\begin{theorem}
\label{thm:MSEgeneralform}
Let $K_{\mb{X}|\mb{Y_i}}$, $i \in \{1,2\}$ be diagonal matrices. Then the rate distortion function of $\RDmse$, $R^{\mse}(\mb{D})$, can be written as
\begin{align*}
R^{\mse}(\mb{D}) = \max\{R^{\mse}_1(\mb{D}), R^{\mse}_2(\mb{D})\}, 
\end{align*}
where
\begin{align}
\label{eq:MSE_1}
R^{\mse}_1(\mb{D}) & = \frac{1}{2}\log\frac{|K_{\mb{X}|\mb{Y_1}}|}{|(D_1)_{l_1}|| \min\{[{D}_1]_{l_2},[\wt{D}_2]_{l_2} \} |}
 + \frac{1}{2}\log\frac{|(\wh{D}_1)_{l_1}|}{|\min\{(\wh{D}_1)_{l_1},({D}_2)_{l_1} \}|} 
 \\
 \label{eq:MSE_2}
R^{\mse}_2(\mb{D}) & = \frac{1}{2}\log\frac{|K_{\mb{X}|\mb{Y_2}}|}{|[D_2]_{l_2}|| \min\{(\wh{D}_1)_{l_1},({D}_2)_{l_1} \} |}
 + \frac{1}{2}\log\frac{|[\wt{D}_2]_{l_2} |}{|\min\{[{D}_1]_{l_2},[\wt{D}_2]_{l_2} \} |},
\end{align}
and\footnote{Note that $\wh{D}_1$ and $\wt{D}_2$ are positive definite since $D^{-1}_1 \succeq K^{-1}_{\mb{X}|\mb{Y_1}} \succ 0$,  $D^{-1}_2 \succeq K^{-1}_{\mb{X}|\mb{Y_2}} \succ 0$, and $K^{-1}_{\mb{X}|\mb{Y_2}} = K^{-1}_{\mb{X}|\mb{Y_1}} +K$.}  $\wh{D}_1 = (D^{-1}_1 + K)^{-1}$, $\wt{D}_2 = (D^{-1}_2 - K)^{-1}$.
\end{theorem}

To prove Theorem \ref{thm:MSEgeneralform}, first we find an upper bound based on the achievable scheme in
 \cite{timo} in Section \ref{sec:ach_scheme} and then we utilize  
  the Enhanced-$\DSM$ bound in the previous section, which turns out
   to match the upper bound.
 
 \begin{remark}
Theorem \ref{thm:MSEgeneralform}
subsumes the Gaussian version of Watanabe's result \cite{watanabe}
by allowing for $\mb{X}$ to have dimension exceeding two.
Watanabe points out that the rate-distortion for his problem, and thus
for ours, does not in general equal the sum of the individual rate-distortion
functions across the components of $\mb{X}$, even though they are
independent, independent given either side information vector, and
subject to separate distortion constraints. Thus, even in this
case, it is necessary to code across the different components of
$\mb{X}$.
 \end{remark}
\begin{theorem}
\label{thm:SI}
 The rate-distortion function for $\RDSI$, $R^{\si}(\mb{D})$,  can be expressed as
\begin{align*}
R^{\si}(\mb{D}) = \max\{R^{\si}_1(\mb{D}), R^{\si}_2(\mb{D})\}, 
\end{align*}
where
\begin{align}
\label{eq:SI_1}
R^{\si}_1(\mb{D}) & = \frac{1}{2}\log\frac{|K_{\mb{X}|\mb{Y_1}}|}{|(D_1)_{l_1}|| \min\{[{D}_1]_{l_2},[\wt{D}_2]_{l_2} \} |}
 + \frac{1}{2}\log\frac{|(\wh{D}_1)_{l_1}|}{|\min\{(\wh{D}_1)_{l_1},({D}_2)_{l_1} \}|} 
 \\
 \label{eq:SI_2}
R^{\si}_2(\mb{D}) & = \frac{1}{2}\log\frac{|K_{\mb{X}|\mb{Y_2}}|}{|[D_2]_{l_2}|| \min\{(\wh{D}_1)_{l_1},({D}_2)_{l_1} \} |}
 + \frac{1}{2}\log\frac{|[\wt{D}_2]_{l_2} |}{|\min\{[{D}_1]_{l_2},[\wt{D}_2]_{l_2} \} |},
\end{align}
and\footnote{Note that $\wh{D}_1$ and $\wt{D}_2$ are positive definite due to similar reasoning as in Theorem \ref{thm:MSEgeneralform}.} $\wh{D}_1 = (D^{-1}_1 + K)^{-1}$, $\wt{D}_2 = (D^{-1}_2 - K)^{-1}$.
\end{theorem}
For the direct part of the proof of Theorem \ref{thm:SI},  we utilize the achievable scheme in Section \ref{sec:ach_scheme}. For the converse result presented in Section \ref{sec:converse}, we use the Enhanced-$\DSM$ bound.
 
\begin{theorem}
\label{thm:trace}
 The rate distortion function for $\RDTr$, $R^{\tr}(\mb{D})$, can be characterized as
 \begin{align*}
 R^{\tr}(\mb{D}) =& \min_{C^{\tr}(\mb{D})}\max \{ R^{\tr}_1(\mb{D}), R^{\tr}_2(\mb{D})
 \} 
 \end{align*}
 where
\begin{align}
\label{eq:TR_1}
R^{\tr}_1(\mb{D}) = &\frac{1}{2}\log\frac{|K_{\mb{X}|\mb{Y_1}}|}{| I +A(K_{\mb{X}|\mb{W},\mb{Y_1}})_{l_1}|}
+ \frac{1}{2}\log\frac{1}{|(K_{\mb{X}|\mb{W},\mb{V},\mb{Y_2}})_{l_1}||[K_{\mb{X}|\mb{W},\mb{U},\mb{Y_1}}]_{l_2}|}, 
\\
\label{eq:TR_2}
R^{\tr}_2(\mb{D}) = &\frac{1}{2}\log\frac{|K_{\mb{X}|\mb{Y_2}}|}{| I - B[K_{\mb{X}|\mb{W},\mb{Y_2}}]_{l_2}|} 
+ \frac{1}{2}\log\frac{1}{|(K_{\mb{X}|\mb{W},\mb{V},\mb{Y_2}})_{l_1}||[K_{\mb{X}|\mb{W},\mb{U},\mb{Y_1}}]_{l_2}|}
\end{align}
and $C^{\tr}(\mb{D})$  denotes
\begin{align}
\label{set1:c1}
&(\mb{W},\mb{U},\mb{V}) \mbox{ jointly Gaussian with } (\mb{X},\mb{Y_1},\mb{Y_2})
\\
\label{set1:c2}
& (\mb{W},\mb{U},\mb{V}) \lra \mb{X} \lra (\mb{Y_1},\mb{Y_2})
 \\
 \label{set1:c3}
& {\mb{U}} \perp (\mb{X})_{l_1}|({\mb{W}}, \mb{Y_1}), \ {\mb{V}} \perp [\mb{X}]_{l_2}|({\mb{W}},\mb{Y_2})
\\
\label{set1:c4}
&K_{\mb{X}|\mb{W},\mb{Y_1}},K_{\mb{X}|\mb{W},\mb{Y_2}},K_{\mb{X}|\mb{W},\mb{U},\mb{Y_1}},K_{\mb{X}|\mb{W},\mb{V},\mb{Y_2}}, \mbox{ diagonal}
 \\
 \label{set1:c5}
&\Tr((K_{\mb{X}|\mb{W},\mb{Y_1}})_{l_1}) + \Tr([K_{\mb{X}|\mb{W},\mb{U},\mb{Y_1}}]_{l_2}) \le D_1 \\
\label{set1:c6}
&\Tr((K_{\mb{X}|\mb{W},\mb{V},\mb{Y_2}})_{l_1})+\Tr([K_{\mb{X}|\mb{W},\mb{Y_2}}]_{l_2}) \le D_2 
\end{align}
\end{theorem}

\begin{remark}
\label{remark:diag}
Let $\mb{W}$ be jointly Gaussian with $(\mb{X},\mb{Y_1},\mb{Y_2})$ such that $\mb{W} \lra \mb{X} \lra (\mb{Y_1},\mb{Y_2})$. Due to (\ref{eq:relation}), $K_{\mb{X}|\mb{W},\mb{Y_1}}$ is a diagonal matrix if and only if $K_{\mb{X}|\mb{W},\mb{Y_2}}$ is  diagonal. 
\end{remark}

Similar to the proof of Theorem \ref{thm:SI} and \ref{thm:MSEgeneralform}, we begin with proving the direct part using the same achievable scheme for $\RDSI$  by changing the distortion measure. For the converse part; however, we utilize the $\mLB$ bound, which is stronger than the Enhanced-$\DSM$ bound in general. 

\section{Achievable Scheme}
\label{sec:ach_scheme}
Heegard and Berger~\cite{berger} give an achievable scheme for
a more general version of our problem. For more than two decoders,
the Heegard and Berger result was corrected by 
Timo \textit{et al.}~\cite{timo}, but we shall only consider the
two-decoder version here. Particularizing the Heegard-Berger result
to our problem implies the following.

\begin{theorem}[cf.~\cite{berger,timo}]
\label{thm:achievable}
The rate distortion function, $R(\mb{D})$, is upper bounded by
\begin{align}
\label{upperb}
R_{ach}(\mb{D}) = \inf_{C_u(\mb{D})} \max \{I(\mb{X};\mb{W}, \mb{U}| \mb{Y_1}) +  I(\mb{X};\mb{V}|\mb{W}, \mb{Y_2}),  I(\mb{X};\mb{W},\mb{V}|\mb{Y_2})+ I(\mb{X};\mb{U}|\mb{W}, \mb{Y_1})  \}
\end{align}
where 
\begin{align*}
&C_u(\mb{D}) : \mbox{ set of } (\mb{W}, \mb{U}, \mb{V}) \mbox{ such that}  
\\
&(\mb{W},\mb{U},\mb{V}) \text{ jointly Gaussian with } (\mb{X},\mb{Y_1},\mb{Y_2}) 
\\
& (\mb{W},\mb{U},\mb{V}) \lra \mb{X} \lra (\mb{Y_1}, \mb{Y_2})  
 \\
& \Gamma_1\left(K_{\mb{X}|\mb{W},\mb{U},\mb{Y_1}}\right) \preceq D_1,\Gamma_2\left( K_{\mb{X}|\mb{W},\mb{V},\mb{Y_2}}\right) \preceq D_2,
\end{align*}
and $\Gamma_i$ can be equal to one of the mappings in (\ref{dm_MSE}), (\ref{dm_SI}), and (\ref{dm_trace}) and the corresponding distortion constraints are as in (\ref{MSE_distcons}), (\ref{si_distcons}),  and (\ref{tr_distcons}) respectively.
\end{theorem}

Here $\mb{W}$ can be viewed as a common message to both decoders, and $\mb{U}$ and $\mb{V}$ are private messages for Decoder $1$ and $2$ respectively. The encoder first creates $\mb{W}$ via vector quantization with a given Gaussian test channel and then generates $\mb{U}$ and $\mb{V}$ with respect to the source and $\mb{W}$. Then $\mb{W}$ is sent to both decoders and $\mb{U}$ and $\mb{V}$ are sent to Decoder $1$ and Decoder $2$, respectively. At the Decoder side, Decoder $1$ decodes $\mb{W}$ and $\mb{U}$ by using its side information $\mb{Y_1}$. Similarly, Decoder $2$ decodes $\mb{W}$ and $\mb{V}$ using $\mb{Y_2}$.

Heegard and Berger do not require $(\mb{W},\mb{U},\mb{V})$ to be jointly
Gaussian with $(\mb{X},\mb{Y_1},\mb{Y_2})$, but we shall only apply
Theorem~\ref{thm:achievable} with $(\mb{W},\mb{U},\mb{V})$  of this form,
so we have added it as a constraint in the statement of the result.
Note that when $(\mb{W},\mb{U},\mb{V})$ are jointly Gaussian 
with $(\mb{X},\mb{Y_1},\mb{Y_2})$ in
(\ref{upperb}), we can write $R_{ach}(\mb{D})$ as
\begin{align*}
R_{ach}(\mb{D}) = \inf_{C_u(\mb{D})} \max \{R_1, R_2 \}
\end{align*}
where
\begin{align}
\label{eq:ub1}
R_1 = & \frac{1}{2}\log\frac{|K_{\mb{X}|\mb{Y_1}}|}{|K_{\mb{X}|\mb{W},\mb{U},\mb{Y_1}}|}\frac{|K_{\mb{X}|\mb{W},\mb{Y_2}}|}{|K_{\mb{X}|\mb{W},\mb{V},\mb{Y_2}}|},
\\
\label{eq:ub2}
R_2 = & \frac{1}{2}\log\frac{|K_{\mb{X}|\mb{Y_2}}|}{|K_{\mb{X}|\mb{W},\mb{V},\mb{Y_2}}|}\frac{|K_{\mb{X}|\mb{W},\mb{Y_1}}|}{|K_{\mb{X}|\mb{W},\mb{U},\mb{Y_1}}|}.
\end{align}

To get an explicit expression for the upper bounds we need to specify the auxiliary random variables more explicitly. The next three propositions give an explicit upper bound on the $R^{\mse}(\mb{D})$, $R^{\si}(\mb{D})$,  and properties of $(\mb{W},\mb{U},\mb{V})$ in the optimizing set $C_u(\mb{D})$ for trace distortion constraints.

\begin{proposition}
 \label{prop:ach_mse}
 $R^{\mse}(\mb{D})$ is upper bounded by 
 \begin{align*}
R^{\mse}_u(\mb{D}) =  \max \{ R_1^{\mse}(\mb{D}), R_2^{\mse}(\mb{D})\}
 \end{align*}
 where $R_1^{\mse}(\mb{D})$ and $R_2^{\mse}(\mb{D})$ are as in (\ref{eq:MSE_1}) and (\ref{eq:MSE_2}) respectively.
 \end{proposition}
 \begin{proof}
 We start the proof by showing that 
\begin{align*}
&G = \left( \begin{array}{c c} (\wh{D}_1)_{l_1} & \mf{0} \\ \mf{0} & [D_2]_{l_2} \end{array} \right),
\end{align*}
where $\wh{D}_1$ as in Theorem \ref{thm:MSEgeneralform}, is dominated by $K_{\mb{X}|\mb{Y_2}}$. 
Since $K_{\mb{X}|\mb{Y_2}}$ and $G$ are diagonal matrices and $D_2 \preceq K_{\mb{X}|\mb{Y_2}}$, it is enough to show that $(\wh{D}_1)_{l_1} \preceq (K_{\mb{X}|\mb{Y_2}})_{l_1}$. Note that $\wh{D}_1 = (D^{-1}_1 + K)^{-1} \preceq K_{\mb{X}|\mb{Y_2}}$ since $D_1 \preceq K_{\mb{X}|\mb{Y_1}}$. Thus, $(\wh{D}_1)_{l_1} \preceq (K_{\mb{X}|\mb{Y_2}})_{l_1}$ and $G \preceq K_{\mb{X}|\mb{Y_2}}$. Then we can select $\mb{W}$ such that it is jointly Gaussian with $\mb{X}$ and $K_{\mb{X}|\mb{W}, \mb{Y_2}} = G$. This implies 
\begin{align*}
K_{\mb{X}|\mb{W},\mb{Y_1}} &=(K^{-1}_{\mb{X}|\mb{W},\mb{Y_2}} - K)^{-1} 
 \\
  &=(G^{-1} - K)^{-1} 
 \\
&=\left( \begin{array}{c c} (D_1)_{l_1} & \mf{0} \\ \mf{0} & [\wt{D}_2]_{l_2} \end{array} \right),
\end{align*}
where $\wt{D}_2$ is as in Theorem \ref{thm:MSEgeneralform}.

Lastly, we select $\mb{U}$ and $\mb{V}$ jointly Gaussian with $\mb{X}$ and $\mb{W}$ such that
\begin{align*}
&K_{\mb{X}|\mb{W},\mb{V},\mb{Y_2}} =  \left( \begin{array}{c c} \min\{(\wh{D}_1)_{l_1},({D}_2)_{l_1} \}& \mf{0} \\ \mf{0} & [D_2]_{l_2} \end{array} \right), \\
&K_{\mb{X}|\mb{W},\mb{U},\mb{Y_1}} =  \left( \begin{array}{c c} ({D}_1)_{l_1}& \mf{0} \\ \mf{0} & \min\{[{D}_1]_{l_2}, [\wt{D}_2]_{l_2} \} \end{array} \right),
\end{align*} 
satisfy the distortion constraints. Evaluating  $R_1$ and $R_2$ for this choice of $(\mb{W},\mb{U},\mb{V})$ gives us  $R_1^{\mse}(\mb{D})$ and $R_2^{\mse}(\mb{D})$.
 \end{proof}
 
 From the selection of the ``common" and ``private" messages, we can make the following observation. The
``common" message is used to hit the distortion constraint of each decoder with equality over the region in which it is ``weaker." We shall apply this strategy in all three problems, in fact.  Note that each decoder may undershoot its distortion constraint over the region in which it is ``stronger" depending on $D_1$, $D_2$ and $K$. Now we provide the following proposition which gives an explicit upper bound on $R^{\si}(\mb{D})$.
 
 \begin{proposition}
 \label{prop:ach_SI}
 $R^{\si}(\mb{D})$ is upper bounded by 
 \begin{align*}
R^{\si}_u(\mb{D}) =  \max \{ R_1^{\si}(\mb{D}), R_2^{\si}(\mb{D})\}
 \end{align*}
 where $R_1^{\si}(\mb{D})$ and $R_2^{\si}(\mb{D})$ are as in (\ref{eq:SI_1}) and (\ref{eq:SI_2}) respectively.
 \end{proposition}
 \begin{proof}
We follow similar approach in the proof of Proposition  \ref{prop:ach_mse}. We take a particular feasible choice of $(\mb{W},\mb{U},\mb{V})$ in $R_{ach}(\mb{D})$ to get an explicit upper bound on the rate-distortion function, $R^{\si}(\mb{D})$. We would like to choose $\mb{W}$ jointly Gaussian with $\mb{X}$ so that $K_{\mb{X}|\mb{W},\mb{Y_2}}$ is equal to 
\begin{align*}
&G = \left( \begin{array}{c c} (\wh{D}_1)_{l_1} & \mf{0} \\ \mf{0} & [D_2]_{l_2} \end{array} \right).
\end{align*}

This is possible if and only if $G$ is dominated by 
$K_{\mb{X}|\mb{Y_2}}$.  To see that this is the case, note that
$K^{-1}_{\mb{X}|\mb{Y_1}} \preceq D^{-1}_1 $ so we 
have $ K^{-1}_{\mb{X}|\mb{Y_1}} + K \preceq D^{-1}_1 + K$, where $K$ is in (\ref{eq:relation}). This implies that
$K^{-1}_{\mb{X}|\mb{Y_2}} \preceq  \wh{D}^{-1}_1$ since
$D^{-1}_1 + K =\wh{D}^{-1}_1$.

Now since $D_1$ and $D_2$ are scaled identity matrices, we must have
either $D_1 \preceq D_2$ or 
$D_1 \succeq D_2$. We shall show that we have 
 $K^{-1}_{\mb{X}|\mb{Y_2}} \preceq G^{-1}$ in both cases.

\textit{Case 1}: $D_1 \preceq D_2$.

Note that $(\wh{D}^{-1}_1)_{l_1} \succeq (D^{-1}_1)_{l_1}\succeq ({D}^{-1}_2)_{l_1}$. Then
\begin{align*}
G^{-1} - D^{-1}_2  
&=\left( \begin{array}{c c} (\wh{D}^{-1}_1)_{l_1} - ({D}^{-1}_2)_{l_1}& \mf{0} \\ \mf{0} & \mf{0} \end{array} \right)  \\
& \succeq 0.
\end{align*}
So $G^{-1} \succeq D^{-1}_2 \succeq K^{-1}_{\mb{X}|\mb{Y_2}}$.

\textit{Case 2}: $D_1 \succeq D_2$.

Note that $[\wh{D}^{-1}_1]_{l_2} \preceq  [D^{-1}_1]_{l_2} \preceq [{D}^{-1}_2]_{l_2}$. 
Then
\begin{align*}
G^{-1}-\wh{D}^{-1}_1  
&=\left( \begin{array}{c c} \mf{0}& \mf{0} \\ \mf{0} & [D^{-1}_2]_{l_2} - [\wh{D}^{-1}_1]_{l_2} \end{array} \right)  \\
& \succeq 0.
\end{align*}
So $G^{-1} \succeq \wh{D}^{-1}_1 \succeq K^{-1}_{\mb{X}|\mb{Y_2}}$.
This shows that $K_{\mb{X}|\mb{Y_2}} \succeq G$ as desired. Hence we can select $K_{\mb{X}|\mb{W},\mb{Y_2}} = G$.

Now for any $\mb{W}$ that is jointly Gaussian with $\mb{X}$ and has the
specified $K_{\mb{X}|\mb{W},\mb{Y_2}}$, we will have
\begin{align*}
K_{\mb{X}|\mb{W},\mb{Y_1}} &=(K^{-1}_{\mb{X}|\mb{W},\mb{Y_2}} - K)^{-1}  \\
 &=\left(\left( \begin{array}{c c} (\wh{D}^{-1}_1)_{l_1} & \mf{0} \\ \mf{0} & [D^{-1}_2]_{l_2} \end{array} \right) - \left( \begin{array}{c c} A & \mf{0} \\ \mf{0} & B \end{array} \right) \right)^{-1}\\
&=\left( \begin{array}{c c} (D_1)_{l_1} & \mf{0} \\ \mf{0} & [\wt{D}_2]_{l_2} \end{array} \right).
\end{align*} 

Then select $\mb{U}$ and $\mb{V}$ jointly Gaussian with $\mb{X}$ and $\mb{W}$ so that
\begin{align*}
&K_{\mb{X}|\mb{W},\mb{V},\mb{Y_2}} =  \left( \begin{array}{c c} \min\{(\wh{D}_1)_{l_1},({D}_2)_{l_1} \}& \mf{0} \\ \mf{0} & [D_2]_{l_2} \end{array} \right), \\
&K_{\mb{X}|\mb{W},\mb{U},\mb{Y_1}} =  \left( \begin{array}{c c} ({D}_1)_{l_1}& \mf{0} \\ \mf{0} & \min\{[{D}_1]_{l_2}, [\wt{D}_2]_{l_2} \} \end{array} \right).
\end{align*} 

Note that $K_{\mb{X}|\mb{W},\mb{U},\mb{Y_1}} \preceq D_1$ and $K_{\mb{X}|\mb{W},\mb{V},\mb{Y_2}} \preceq D_2$ as required.
Evaluating  $R_1$ and $R_2$ for this choice of $(\mb{W},\mb{U},\mb{V})$ gives us  $R_1^{\si}(\mb{D})$ and $R_2^{\si}(\mb{D})$.
\end{proof}

As in the achievable scheme for $\RDmse$ in Proposition \ref{prop:ach_mse}, each decoder hits its own distortion constraint with equality on the region where it is ``weaker" while each may undershoot its distortion constraint where it is ``stronger" depending on $D_1$, $D_2$ and $K$. 
Finally, we provide the following proposition giving additional constraints on the optimizers in the optimization set  $C_u(\mb{D})$ when we have trace distortion constraints.
 \begin{proposition}
 \label{prop:ach_trace}
 $R^{\tr}(\mb{D})$ is upper bounded by 
 \begin{align}
 \label{eq:ach_trace}
 R^{\tr}_u(\mb{D}) =\min_{C^{\tr}(\mb{D})}\max \{ R_1^{\tr}(\mb{D}), R_2^{\tr}(\mb{D})\}
 \end{align}
 where $R_1^{\tr}(\mb{D})$, $R_2^{\tr}(\mb{D})$ and $C^{\tr}(\mb{D})$ as in Theorem \ref{thm:trace}.
 \end{proposition}

\begin{proof}
Notice that we can include the conditions
\begin{align}
\label{cond_1}
& {\mb{U}} \perp (\mb{X})_{l_1}|({\mb{W}}, \mb{Y_1}), \ {\mb{V}} \perp [\mb{X}]_{l_2}|({\mb{W}},\mb{Y_2})
\\
\label{cond_2}
&K_{\mb{X}|\mb{W},\mb{Y_1}},K_{\mb{X}|\mb{W},\mb{Y_2}},K_{\mb{X}|\mb{W},\mb{U},\mb{Y_1}},K_{\mb{X}|\mb{W},\mb{V},\mb{Y_2}} \mbox{ diagonal} 
\end{align}
to $C_u(\mb{D})$ of $R_{ach}(\mb{D})$, which gives the result. 
\end{proof}


\section{Converse Results}
\label{sec:converse}
\subsection{Converse for $\RDmse$ and $\RDSI$ }

It turns out that the Enhancement-$\DSM$ is sufficient for
the $\RDmse$ and $\RDSI$ problems, so we will use that bound.
We shall select $\mb{Y}$ in the Enhancement-$\DSM$ with 
the properties stated in the following lemma.

\begin{lemma}
\label{lemma:Yspec}
Let the joint distribution of the source and side information pairs $(\mb{X}, \mb{Y_i})$, $i \in \{1,2\}$ be given.
We can find a random vector, $\mb{Y}$,  jointly Gaussian with $(\mb{X},\mb{Y_1},\mb{Y_2})$ such that 
\begin{align}
\mb{X} \lra \mb{Y} \lra (\mb{Y_1}, \mb{Y_2})
\end{align}
and
\begin{align}
\label{deg_1}
K^{-1}_{\mb{X}|\mb{Y}} &= K^{-1}_{\mb{X}|\mb{Y_1}} + \wh{K}
\\
\label{deg_2}
& = K^{-1}_{\mb{X}|\mb{Y_2}} + \wt{K} 
\end{align}
where $\wh{K} = \left( \begin{array}{c c} A & \mf{0} \\ \mf{0} & \mf{0}  \end{array} \right)$ and  $\wt{K} = \left( \begin{array}{c c}\mf{0} & \mf{0} \\ \mf{0} & -B  \end{array} \right)$.
\end{lemma}

\begin{proof}
Observe that if $(\mb{X},\mb{Y}, \mb{Y_i})$ can be coupled so that $\mb{X}\lra \mb{Y} \lra \mb{Y_i}$ holds for $i \in \{1,2\}$ and $(\mb{X},\mb{Y})$ has the same distribution under both couplings then it is possible to couple all four variables such that 
$\mb{X} \lra \mb{Y} \lra (\mb{Y_1}, \mb{Y_2})$ holds.

Next note that the matrix $K^{-1}_{\mb{X}|\mb{Y_1}} - K^{-1}_{\mb{X}}+ \wh{K} = K^{-1}_{\mb{X}|\mb{Y_2}} - K^{-1}_{\mb{X}}+ \wt{K}$ is positive semidefinite. Thus, we can find a matrix $M$ such that $M^TM = K^{-1}_{\mb{X}|\mb{Y_1}} - K^{-1}_{\mb{X}}+ \wh{K}= K^{-1}_{\mb{X}|\mb{Y_2}} - K^{-1}_{\mb{X}}+ \wt{K}$. Then, let $\mb{N}$ be a Gaussian random vector, independent of $\mb{X}$, with covariance matrix $K_{\mb{N}} = I$ and let $\mb{Y} = M\mb{X} + \mb{N}$. Then, 
$K^{-1}_{\mb{X}|\mb{Y}} = K^{-1}_{\mb{X}} + M^TM = K^{-1}_{\mb{X}|\mb{Y_1}}+ \wh{K} = K^{-1}_{\mb{X}|\mb{Y_2}} + \wt{K}$. Since we have $K_{\mb{X}|\mb{Y}} \preceq K_{\mb{X}|\mb{Y_i}}$, $i \in \{1,2\}$, we can couple  $(\mb{X},\mb{Y}, \mb{Y_i})$ so that $\mb{X}\lra \mb{Y} \lra \mb{Y_i}$. 
\end{proof}

Let $\mb{Y}$ be selected as in Lemma \ref{lemma:Yspec}. 
By Lemma~\ref{lemma:gauss_ELB} we can  add the condition that $(\mb{W},\mb{U},\mb{V})$ is jointly Gaussian with  the source  and  side information  at decoders to optimization sets $\bar{C}_{l1}$ and  $\bar{C}_{l2}$ in the Enhanced-$\DSM$. Then we can write $R_{lo1}$ in (\ref{eq:ELB}) as
\begin{align*}
R_{lo1} = \frac{1}{2}\log\frac{|K_{\mb{X}|\mb{Y_1}}|}{|K_{\mb{X}|\mb{W},\mb{U},\mb{Y_1}}|}\frac{|K_{\mb{X}|\mb{W},\mb{U},\mb{Y}}|}{|K_{\mb{X}|\mb{W},\mb{U},\mb{V},\mb{Y}}|}.
\end{align*}

Likewise, $R_{lo2}$ in (\ref{eq:ELB}) can be written as
\begin{align*}
R_{lo2} = \frac{1}{2}\log\frac{|K_{\mb{X}|\mb{Y_2}}|}{|K_{\mb{X}|\mb{W},\mb{V},\mb{Y_2}}|}\frac{|K_{\mb{X}|\mb{W},\mb{V},\mb{Y}}|}{|K_{\mb{X}|\mb{W},\mb{U},\mb{V},\mb{Y}}|}.
\end{align*}

We can further write,
 \begin{align}
R_{lo1} &=\frac{1}{2}\log\frac{|K_{\mb{X}|\mb{Y_1}}|}{|K^{-1}_{\mb{X}|\mb{W},\mb{U},\mb{Y_1}} + \wh{K} |}\frac{|K^{-1}_{\mb{X}|\mb{W},\mb{U},\mb{Y_1}}|}{|K_{\mb{X}|\mb{W},\mb{U},\mb{V},\mb{Y}}|}
\notag \\
& = \frac{1}{2}\log\frac{|K_{\mb{X}|\mb{Y_1}}|}{|I + \wh{K}K_{\mb{X}|\mb{W},\mb{U},\mb{Y_1}}  |}\frac{1}{|K_{\mb{X}|\mb{W},\mb{U},\mb{V},\mb{Y}}|} 
\notag \\
\label{ineq:exp_exprELB}
&\ge \frac{1}{2}\log\frac{|K_{\mb{X}|\mb{Y_1}}|}{\prod^{l_1 + l_2} _{i=1}(1 + (\wh{K})_{ii}(K_{\mb{X}|\mb{W},\mb{U},\mb{Y_1}})_{ii})  } 
\frac{1}{\prod^{l_1 + l_2} _{i=1}(K_{\mb{X}|\mb{W},\mb{U},\mb{V},\mb{Y}})_{ii}}.
\end{align}
 
 Now we focus on $\RDmse$ where $K_{\mb{X}|\mb{Y_i}}$, $i \in \{1,2\}$ are diagonal matrices and $D_i$, $i \in \{1,2\}$ are as in (\ref{MSE_distcons}). Since $(\mb{W},\mb{U},\mb{V})$ is jointly Gaussian with $(\mb{X},\mb{Y_1},\mb{Y_2},\mb{Y})$, we can write $K^{-1}_{\mb{X}| \mb{W},\mb{U}, \mb{V} ,\mb{Y_2}} = K^{-1}_{\mb{X}| \mb{W},\mb{U}, \mb{V} ,\mb{Y}} - \wh{K}$, where $\wh{K}$ as in Lemma \ref{lemma:Yspec}.  Then we can write  $(K_{\mb{X}| \mb{W},\mb{U}, \mb{Y_1}})_{diag} \preceq D_1$  and $((K^{-1}_{\mb{X}| \mb{W},\mb{U}, \mb{V} ,\mb{Y}} - \wh{K})^{-1})_{diag} \preceq D_2$, the constraints  at  $\bar{C}_{l1}$, as $(K_{\mb{X}| \mb{W},\mb{U}, \mb{Y_1}})_{diag} \preceq D_1$  and $(K_{\mb{X}| \mb{W},\mb{U}, \mb{V} ,\mb{Y_2}})_{diag} \preceq D_2$.
 
 The following lemma will be useful for matching the distortion constraints in the achievable scheme and the Enhanced-$\DSM$.
 
 \begin{lemma}
\label{lemma:matrixineq}
Let $A \succeq 0$ be an $m \times m$ diagonal matrix, $M \succ 0$ be an $m \times m$ matrix and $M_{diag} $ denote $(M)_{diag}$.
Then $[(M_{diag})^{-1} + A]^{-1} \succeq ([M^{-1} + A]^{-1})_{diag}$.
\end{lemma}
\begin{proof}
See Appendix \ref{app:matrix_ineq}.
\end{proof}
 
From  $(K_{\mb{X}| \mb{W},\mb{U}, \mb{Y_1}})_{diag} \preceq D_1$  and $(K_{\mb{X}| \mb{W},\mb{U}, \mb{V} ,\mb{Y_2}})_{diag} \preceq D_2$, the constraints in $\bar{C}_{l1}$, and  by Lemma \ref{lemma:matrixineq} we can get 
\begin{align*}
&(K_{\mb{X}| \mb{W},\mb{U}, \mb{Y}})_{diag} \preceq (D^{-1}_1 + \wh{K})^{-1}
\\
& 
(K_{\mb{X}| \mb{W},\mb{U}, \mb{V} ,\mb{Y}})_{diag} \preceq (D^{-1}_2 + \wt{K})^{-1}
\end{align*}
 which implies 
\begin{align*}
&(K_{\mb{X}| \mb{W},\mb{U}, \mb{V} ,\mb{Y}})_{diag} \preceq \min((D^{-1}_1 + \wh{K})^{-1}, (D^{-1}_2 + \wt{K})^{-1}).
\end{align*} 

Let  $\wh{D}_1$ and $\wt{D}_2$ be as in Theorem \ref{thm:MSEgeneralform}. Note that
$ ((D^{-1}_1 + \wh{K})^{-1})_{l_1} = (\wh{D}_1)_{l_1} $ and  $[(D^{-1}_1 + \wh{K})^{-1}]_{l_2}  = [D_1]_{l_2}$ Also,  $((D^{-1}_2 + \wt{K})^{-1})_{l_1} = (D_2)_{l_1}$ and $[(D^{-1}_2 + \wt{K})^{-1}]_{l_2} = [\wt{D}_2]_{l_2} $. Then the  right hand side of (\ref{ineq:exp_exprELB}) is lower bounded by
  \begin{align*}
\frac{1}{2}\log\frac{|K_{\mb{X}|\mb{Y_1}}|}{|I + A(D_1)_{l_1}  |} 
\frac{1}{|\min((\wh{D}_1)_{l_1}, ({D_2})_{l_1}) |} 
\frac{1}{|\min([{D_1}]_{l_2}, [{\wt{D}_2}]_{l_2})|}.
\end{align*}

Since
\begin{align*}
\frac{1}{2}\log\frac{|K_{\mb{X}|\mb{Y_1}}|}{|I + A(D_1)_{l_1}  |} =
   \frac{1}{2} \log \frac{|K_{\mb{X}|\mb{Y_1}}| \cdot |(\wh{D}_1)_{l_1}|}{|(D_1)_{l_1}|},
\end{align*}
we have $R_{lo1} \ge R^{\mse}_1(\mb{D})$.
If we follow a similar procedure for $R_{lo2}$, we obtain
\begin{align*}
R_{lo2} 
&= \frac{1}{2}\log\frac{|K_{\mb{X}|\mb{Y_2}}|}{|I + \wt{K}K_{\mb{X}|\mb{W},\mb{V},\mb{Y_2}}  |}\frac{1}{|K_{\mb{X}|\mb{W},\mb{U},\mb{V},\mb{Y}}|} 
\\
&\ge \frac{1}{2}\log\frac{|K_{\mb{X}|\mb{Y_2}}|}{\prod^{l_1 + l_2} _{i=1}(1 + (\wt{K})_{ii}(K_{\mb{X}|\mb{W},\mb{V},\mb{Y_2}})_{ii})} 
 \frac{1}{\prod^{l_1 + l_2} _{i=1}(K_{\mb{X}|\mb{W},\mb{U},\mb{V},\mb{Y}})_{ii}}
  \\
&\ge \frac{1}{2}\log\frac{|K_{\mb{X}|\mb{Y_2}}|}{|I -B[D_2]_{l_2}|} 
 \frac{1}{|\min((\wh{D}_1)_{l_1}, ({D_2})_{l_1}) |} 
 \frac{1}{|\min([{D_1}]_{l_2}, [{\wt{D}_2}]_{l_2})|}.
\end{align*}

Since $\frac{1}{2} \log \frac{|K_{\mb{X}|\mb{Y_2}}|}{|I - B [D_2]_{l_2}|} 
    = \frac{1}{2} \log \frac{|K_{\mb{X}|\mb{Y_2}}| \cdot |[\wt{D}_2]_{l_2}|}{|[D_2]_{l_2}|}$, we have $R_{lo2} \ge R^{\mse}_2(\mb{D})$.
 Hence together with Proposition \ref{prop:ach_mse}, this proves Theorem \ref{thm:MSEgeneralform}. 

Note that for $\RDSI$ we can lower bound the right hand side of (\ref{ineq:exp_exprELB}) by
\begin{align*}
\frac{1}{2}\log\frac{|K_{\mb{X}|\mb{Y_1}}|}{|I + A(D_1)_{l_1}  |} 
\frac{1}{|\min((\wh{D}_1)_{l_1}, ({D_2})_{l_1}) |} 
\frac{1}{|\min([{D_1}]_{l_2}, [{\wt{D}_2}]_{l_2})|}, 
\end{align*}
where $D_i$, $i \in \{1,2 \}$, $\wh{D}_1$ and  $\wt{D}_2$ are as in Theorem \ref{thm:SI}. 
Since
\begin{align*}
\frac{1}{2}\log\frac{|K_{\mb{X}|\mb{Y_1}}|}{|I + A(D_1)_{l_1}  |} =
   \frac{1}{2} \log \frac{|K_{\mb{X}|\mb{Y_1}}| \cdot |(\wh{D}_1)_{l_1}|}{|(D_1)_{l_1}|},
\end{align*}
we have $R_{lo1} \ge R^{\si}_1(\mb{D})$.
If we follow a similar procedure for $R_{lo2}$, we obtain
\begin{align*}
R_{lo2} 
&= \frac{1}{2}\log\frac{|K_{\mb{X}|\mb{Y_2}}|}{|I + \wt{K}K_{\mb{X}|\mb{W},\mb{V},\mb{Y_2}}  |}\frac{1}{|K_{\mb{X}|\mb{W},\mb{U},\mb{V},\mb{Y}}|} 
\\
&\ge \frac{1}{2}\log\frac{|K_{\mb{X}|\mb{Y_2}}|}{\prod^{l_1 + l_2} _{i=1}(1 + (\wt{K})_{ii}(K_{\mb{X}|\mb{W},\mb{V},\mb{Y_2}})_{ii})} 
 \frac{1}{\prod^{l_1 + l_2} _{i=1}(K_{\mb{X}|\mb{W},\mb{U},\mb{V},\mb{Y}})_{ii}}
  \\
&\ge \frac{1}{2}\log\frac{|K_{\mb{X}|\mb{Y_2}}|}{|I -B[D_2]_{l_2}|} 
 \frac{1}{|\min((\wh{D}_1)_{l_1}, ({D_2})_{l_1}) |} 
 \frac{1}{|\min([{D_1}]_{l_2}, [{\wt{D}_2}]_{l_2})|}.
\end{align*}

Since $\frac{1}{2} \log \frac{|K_{\mb{X}|\mb{Y_2}}|}{|I - B [D_2]_{l_2}|} 
    = \frac{1}{2} \log \frac{|K_{\mb{X}|\mb{Y_2}}| \cdot |[\wt{D}_2]_{l_2}|}{|[D_2]_{l_2}|}$, we have $\bar{R}_{lo2} \ge R^{\si}_2(\mb{D})$.
 Hence together with Proposition \ref{prop:ach_SI}, this proves Theorem \ref{thm:SI}.
 
\subsection{ Converse for $\RDTr$}

For $\RDTr$, we utilize the $\mLB$. Similar to the converse of $\RDmse$ and $\RDSI$, 
let $\mb{Y}$ in $\mLB$ be selected as in Lemma \ref{lemma:Yspec}. Then, by Lemma \ref{lemma:hatZ} in Appendix \ref{app:ELB} we can create a  $\mb{\wh{Y}_i}$, $i \in \{1,2\}$ so that $(\mb{X}, \mb{Y}, \mb{Y_1}, \mb{\wh{Y}_i})$ is jointly Gaussian, $\mb{\wh{Y}_i} \lra  \mb{X}\lra \mb{Y_i}$ and $E[\mb{X} | \mb{Y_i}, \mb{\wh{Y}_i}] = E[\mb{X} | \mb{Y_i},\mb{Y}]$ almost surely.
Since $\mb{\wh{Y}_i} \lra  \mb{X}\lra \mb{Y_i}$, we can write 
\begin{align*}
\wh{\mb{Y}}_\mb{i} = A_{\wh{\mb{Y}}_\mb{i}}\mb{X} +\mb{N}_{\wh{\mb{Y}}_\mb{i}}, i \in \{1,2\},
\end{align*}
where $\mb{N}_{\wh{\mb{Y}}_\mb{i}}$ is independent of $\mb{X}$ and $\mb{Y_i}$.

Then, 
\begin{align}
\label{ineq:new}
&K^{-1}_{\mb{X}|\wh{\mb{Y}}_\mb{i},\mb{Y_i}} = K^{-1}_{\mb{X}|\mb{Y_i}} + A^T_{\wh{\mb{Y}}_\mb{i}}K^{-1}_{\mb{N}_{\wh{\mb{Y}}_\mb{i}}}A_{\wh{\mb{Y}}_\mb{i}}.
\end{align}

Also, since $E[\mb{X} | \mb{Y_i}, \mb{\wh{Y}_i}] = E[\mb{X} | \mb{Y_i},\mb{Y}]$ almost surely 
$K^{-1}_{\mb{X}|\mb{Y}}= K^{-1}_{\mb{X}|\mb{Y},\mb{Y_i}} = K^{-1}_{\mb{X}|\wh{\mb{Y}}_\mb{i},\mb{Y_i}}$.
Then, from (\ref{ineq:new}),  $K^{-1}_{\mb{X}|\mb{Y}}- K^{-1}_{\mb{X}|\mb{Y_1}} = \wh{K}$ and  $K^{-1}_{\mb{X}|\mb{Y}}- K^{-1}_{\mb{X}|\mb{Y_2}} = \wt{K}$, 
\begin{align}
\label{eq:hat_Y1}
A^T_{\wh{\mb{Y}}_\mb{1}}K^{-1}_{\mb{N}_{\wh{\mb{Y}}_\mb{1}}}A_{\wh{\mb{Y}}_\mb{1}} = \wh{K}.
\\
\label{eq:hat_Y2}
A^T_{\wh{\mb{Y}}_\mb{2}}K^{-1}_{\mb{N}_{\wh{\mb{Y}}_\mb{2}}}A_{\wh{\mb{Y}}_\mb{2}} = \wt{K}.
\end{align}

Now, we consider any feasible variable satisfying the constraints in the optimization of $R_{lo}(\mb{D})$  in Theorem \ref{thm:minmax}. We can rewrite $R_{lo1}$ in (\ref{rlo1}) as
\begin{align}
R_{lo1} &= I(\mb{X};\mb{W},\mb{U}|\mb{Y_1}) + I(\mb{X};\mb{V}|\mb{W},\mb{U},\mb{Y}) 
\notag \\
&= h(\mb{X}|\mb{Y_1}) -h(\mb{X}|\mb{W},\mb{U},\mb{Y_1}) + h(\mb{X}|\mb{W},\mb{U},\mb{Y}) 
- h(\mb{X}|\mb{W},\mb{U},\mb{V},\mb{Y})
\notag \\
& =  h(\mb{X}|\mb{Y_1}) -h(\mb{X}|\mb{W},\mb{U},\mb{Y_1}) + h(\mb{X}|\mb{W},\mb{U},\mb{Y_1},\mb{Y}) 
- h(\mb{X}|\mb{W},\mb{U},\mb{V},\mb{Y}) \mbox{ since }
\mb{X} \lra \mb{Y} \lra \mb{Y_1}.
\notag 
\end{align}

Since $\mb{X} \lra E[\mb{X}| \mb{Y_1},\mb{Y}] \lra (\mb{Y_1},\mb{Y})$ and  $\mb{X} \lra  (\mb{Y_1}, \mb{Y}) \lra E[\mb{X}| \mb{Y_1},\mb{Y}]$, $h(\mb{X}|\mb{W},\mb{U},\mb{Y_1},\mb{Y}) = h(\mb{X}|\mb{W},\mb{U},E[\mb{X}| \mb{Y}, \mb{Y_1}] )$.
Furthermore, we can write $h(\mb{X}|\mb{W},\mb{U},E[\mb{X}| \mb{Y}, \mb{Y_1}] )
= h(\mb{X}|\mb{W},\mb{U},E[\mb{X}| \mb{Y_1}, \mb{\wh{Y}_1}] )$, since $E[\mb{X} | \mb{Y_1}, \mb{\wh{Y}_1}] = E[\mb{X} | \mb{Y_1},\mb{Y}]$ almost surely. Then
we can write
\begin{align}
R_{lo1}& = h(\mb{X}|\mb{Y_1}) -h(\mb{X}|\mb{W},\mb{U},\mb{Y_1}) + h(\mb{X}|\mb{W},\mb{U},\mb{Y_1},\mb{\wh{Y}_1}) 
- h(\mb{X}|\mb{W},\mb{U},\mb{V},\mb{Y}) 
\notag
\\
 &= h(\mb{X}|\mb{Y_1}) -I(\mb{X};\wh{\mb{Y}}_\mb{1}|\mb{W},\mb{U},\mb{Y_1}) - h(\mb{X}|\mb{W},\mb{U},\mb{V},\mb{Y}) \notag
\\
&= h(\mb{X}|\mb{Y_1}) +h(\wh{\mb{Y}}_\mb{1}|\mb{X},\mb{Y_1})- h(\wh{\mb{Y}}_\mb{1}|\mb{W},\mb{U},\mb{Y_1}) \notag
- h(\mb{X}|\mb{W},\mb{U},\mb{V},\mb{Y}) 
\notag \\
\label{eq:lo1_det}
&\ge \frac{1}{2}\log\frac{|K_{\mb{X}|\mb{Y_1}}|}{|K_{\wh{\mb{Y}}_\mb{1}|\mb{W},\mb{U},\mb{Y_1}}|}\frac{|K_{\wh{\mb{Y}}_\mb{1}|\mb{X},\mb{Y_1}}|}{|K_{\mb{X}|\mb{W},\mb{U},\mb{V},\mb{Y}}|} 
\end{align}
 with equality if $(\mb{W},\mb{U},\mb{V})$ is Gaussian achieving the given covariance matrices.  
Now, let us focus on the ratio $\frac{|K_{\wh{\mb{Y}}_\mb{1}|\mb{X},\mb{Y_1}}|}{|K_{\wh{\mb{\mb{Y}}}_\mb{1}|\mb{W},\mb{U},\mb{Y_1}}|}$. Since $\wh{\mb{\mb{Y}}}_\mb{1} \lra \mb{X} \lra \mb{Y_1}$ we can write
\begin{align*}
&\frac{|K_{\wh{\mb{\mb{Y}}}_\mb{1}|\mb{X},\mb{Y_1}}|}{|K_{\wh{\mb{\mb{Y}}}_\mb{1}|\mb{W},\mb{U},\mb{Y_1}}|}=
\frac{|K_{\mb{N}_{\wh{\mb{Y}}_\mb{1}}}|}{|K_{\mb{N}_{\wh{\mb{Y}}_\mb{1}}} +A_{\wh{\mb{\mb{Y}}}_\mb{1}}K_{\mb{X}|\mb{W},\mb{U},\mb{Y_1}}A^{T}_{\wh{\mb{\mb{Y}}}_\mb{1}} |}.
\end{align*}
Since $K_{N_{\wh{\mb{\mb{Y}}}_\mb{1}}}$ is positive definite we can write it as $S_{\wh{\mb{\mb{Y}}}_\mb{1}}S_{\wh{\mb{\mb{Y}}}_\mb{1}}$ where $S_{\wh{\mb{\mb{Y}}}_\mb{1}}$ is an invertible matrix. Then we can write,
\begin{align*}
\frac{|K_{\wh{\mb{\mb{Y}}}_\mb{1}|\mb{X},\mb{Y_1}}|}{|K_{\wh{\mb{\mb{Y}}}_\mb{1}|\mb{W},\mb{U},\mb{Y_1}}|} & =\frac{1}{| I + S^{-1}_{\wh{\mb{\mb{Y}}}_\mb{1}}A_{\wh{\mb{\mb{Y}}}_\mb{1}}K_{\mb{X}|\mb{W},\mb{U},\mb{Y_1}}A^{T}_{\wh{\mb{\mb{Y}}}_\mb{1}}S^{-1}_{\wh{\mb{\mb{Y}}}_\mb{1}} |}
\\
&=\frac{1}{| I + K_{\mb{X}|\mb{W},\mb{U},\mb{Y_1}}A^{T}_{\wh{\mb{\mb{Y}}}_\mb{1}}S^{-1}_{\wh{\mb{\mb{Y}}}_\mb{1}}S^{-1}_{\wh{\mb{\mb{Y}}}_\mb{1}}A_{\wh{\mb{\mb{Y}}}_\mb{1}} |},
 \mbox{ by Sylvester's determinant identity}
\\
&=\frac{1}{| I + K_{\mb{X}|\mb{W},\mb{U},\mb{Y_1}}A^{T}_{\wh{\mb{\mb{Y}}}_\mb{1}}K^{-1}_{\mb{N}_{\wh{\mb{Y}}_\mb{1}}}A_{\wh{\mb{Y}}_\mb{1}} |}
\\
&=\frac{1}{| I + K_{\mb{X}|\mb{W},\mb{U},\mb{Y_1}}\left(\begin{array}{c c}  A&\mf{0}\\ \mf{0}&\mf{0}\end{array} \right) |},
\end{align*}
where the last equality is due to (\ref{eq:hat_Y1}).
Then we can write (\ref{eq:lo1_det}) as
\begin{align*}
R_{lo1} & \ge \frac{1}{2}\log\frac{|K_{\mb{X}|\mb{Y_1}}|}{| I + K_{\mb{X}|\mb{W},\mb{U},\mb{Y_1}}\left(\begin{array}{c c}  A&\mf{0}\\ \mf{0}&\mf{0}\end{array} \right) |}\frac{1}{|K_{\mb{X}|\mb{W},\mb{U},\mb{V},\mb{Y}}|} 
\\
&= \frac{1}{2}\log\frac{|K_{\mb{X}|\mb{Y_1}}|}{|  \left(\begin{array}{c c}  I +(K_{\mb{X}|\mb{W},\mb{U},\mb{Y_1}})_{l_1}A&\mf{0}\\ (K_{[\mb{X}]_{l_2}(\mb{X})_{l_1}|\mb{W},\mb{U},\mb{Y_1}})_{l_1}A & I \end{array} \right) |}
 \frac{1}{|K_{\mb{X}|\mb{W},\mb{U},\mb{V},\mb{Y}}|} 
\\
&= \frac{1}{2}\log\frac{|K_{\mb{X}|\mb{Y_1}}|}{| I + (K_{\mb{X}|\mb{W},\mb{U},\mb{Y_1}})_{l_1}A|}
\frac{1}{|K_{\mb{X}|\mb{W},\mb{U},\mb{V},\mb{Y}}|} 
\\
&\ge \frac{1}{2}\log\frac{|K_{\mb{X}|\mb{Y_1}}|}{\prod^{l_1}_{i=1} (1 + (K_{\mb{X}|\mb{W},\mb{U},\mb{Y_1}})_{ii}(A)_{ii})}
\frac{1}{\prod^{k}_{i=1}(K_{\mb{X}|\mb{W},\mb{U},\mb{V},\mb{Y}})_{ii}}, 
\mbox{ by Hadamard inequality,} 
\end{align*}
with equality if $(K_{\mb{X}|\mb{W},\mb{U},\mb{Y_1}})_{l_1}$ and $K_{\mb{X}|\mb{W},\mb{U},\mb{V},\mb{Y}}$ are diagonal matrices. Since $K_{\mb{X}|\mb{W},\mb{U},\mb{V},\mb{Y}} \preceq K_{\mb{X}|\mb{W},\mb{U},\mb{Y}}$ and  $K_{\mb{X}|\mb{W},\mb{U},\mb{V},\mb{Y}} \preceq K_{\mb{X}|\mb{W},\mb{V},\mb{Y}}$ imply $(K_{\mb{X}|\mb{W},\mb{U},\mb{V},\mb{Y}})_{ii} \le \min \{(K_{\mb{X}|\mb{W},\mb{U},\mb{Y}})_{ii}, (K_{\mb{X}|\mb{W},\mb{V},\mb{Y}})_{ii}\} $ for all $i \in [k]$, we can further write

\begin{align}
\label{eq:lowerb_11}
 &R_{lo1} \ge \frac{1}{2}\log\frac{|K_{\mb{X}|\mb{Y_1}}|}{\prod^{l_1}_{i=1} (1 + (K_{\mb{X}|\mb{W},\mb{U},\mb{Y_1}})_{ii}(A)_{ii})} 
 +\frac{1}{2}\log\frac{1}{\prod^{k}_{i=1}\min \{(K_{\mb{X}|\mb{W},\mb{U},\mb{Y}})_{ii}, (K_{\mb{X}|\mb{W},\mb{V},\mb{Y}})_{ii}\} }. 
\end{align}

By applying the same procedure as above for the $R_{lo2}$ we can get
\begin{align}
\label{eq:lowerb_21}
&R_{lo2} \ge \frac{1}{2}\log\frac{|K_{\mb{X}|\mb{Y_2}}|}{\prod^{l_2}_{i=1} (1 - ([K_{\mb{X}|\mb{W},\mb{V},\mb{Y_2}}]_{l_2})_{ii}(B)_{ii})}
 +\frac{1}{2}\log
\frac{1}{\prod^{k}_{i=1}\min \{(K_{\mb{X}|\mb{W},\mb{U},\mb{Y}})_{ii}, (K_{\mb{X}|\mb{W},\mb{V},\mb{Y}})_{ii}\} }.
\end{align} 

 We denote the right-hand sides of (\ref{eq:lowerb_11}) and (\ref{eq:lowerb_21}) as $\wh{R}_{lo1}$ and $\wh{R}_{lo2}$ respectively.
  The next proposition gives a tight lower bound to $R^{\tr}(\mb{d})$  by specifying the properties of the optimizers in $\mLB$.
 \begin{proposition}
 \label{prop:minmaxtrace}
The rate distortion function of $\RDTr$, $R^{\tr}(\mb{D})$, is lower bounded by 
  \begin{align}
  \label{eq:minimax_lbtr}
& \min_{C^{\tr}(\mb{D})}\max \{ R^{\tr}_1(\mb{D}), R^{\tr}_2(\mb{D})
 \} 
 \end{align}
 where $C^{\tr}(\mb{D})$, $ R^{\tr}_1(\mb{D})$ and $ R^{\tr}_2(\mb{D})$  are as in Theorem \ref{thm:trace}.
 \end{proposition}
 
 The proof  follows from the next four lemmas. At each lemma, we show that without loss of optimality we can add a constraint to the optimization set, $C_l(\mb{D})$ of Theorem \ref{thm:minmax} for the trace constraints. With those additional constraints  $C_l(\mb{D})$ becomes  $C^{\tr}(\mb{D})$ and $\wh{R}_{loi} = R_i^{\tr}(\mb{D})$ for $i \in \{ 1, 2 \}$.
 
 \begin{lemma}
\label{lemma:gauss}
There exists a feasible $(\mb{W_G},\mb{U_G},\mb{V_G})$ for  $R_{lo}(\mb{D})$  such that $(\mb{W_G},\mb{U_G},\mb{V_G})$ are jointly Gaussian with $(\mb{X},\mb{Y},\mb{Y_1},\mb{Y_2})$. Furthermore, such $(\mb{W_G},\mb{U_G},\mb{V_G})$ do not increase $\wh{R}_{lo1}$ and $\wh{R}_{lo2}$.
\end{lemma}

\begin{proof}
Let $(\mb{W_G},\mb{U_G},\mb{V_G})$ be jointly Gaussian with $(\mb{X},\mb{Y},\mb{Y_1},\mb{Y_2})$ and $(\mb{W_G},\mb{U_G},\mb{V_G}) \lra \mb{X} \lra (\mb{Y},\mb{Y_1},\mb{Y_2})$ such that
\begin{align*}
& K_{\mb{X}|\mb{W_G},\mb{U_G},\mb{Y_1}} = K_{\mb{X}|\mb{W},\mb{U},\mb{Y_1}}, \\
&K_{\mb{X}|\mb{W_G},\mb{V_G},\mb{Y_2}}=K_{\mb{X}|\mb{W},\mb{V},\mb{Y_2}}.
\end{align*}

By Lemma \ref{lemma:inequality}, we have $K_{\mb{X}|\mb{W},\mb{U},\mb{Y}}\preceq K_{\mb{X}|\mb{W_G},\mb{U_G},\mb{Y}}$ and $K_{\mb{X}|\mb{W},\mb{V},\mb{Y}}\preceq K_{\mb{X}|\mb{W_G},\mb{V_G},\mb{Y}}$. 
This implies 
\\
$\min \{(K_{\mb{X}|\mb{W},\mb{U},\mb{Y}})_{ii}, (K_{\mb{X}|\mb{W},\mb{V},\mb{Y}})_{ii}\}$ is lower than or equal to $\min \{(K_{\mb{X}|\mb{W_G},\mb{U_G},\mb{Y}})_{ii}, (K_{\mb{X}|\mb{W_G},\mb{V_G},\mb{Y}})_{ii}\}$ for all $i \in [k]$. Hence, we can conclude that  $(\mb{W_G},\mb{U_G},\mb{V_G})$ is feasible for  $R_{lo}(\mb{D})$ and replacing the $(\mb{W},\mb{U},\mb{V})$ with $(\mb{W_G},\mb{U_G},\mb{V_G})$  on (\ref{eq:lowerb_11}) and (\ref{eq:lowerb_21}) does not increase $\wh{R}_{lo1}$ and $\wh{R}_{lo2}$. 
\end{proof}

Then by Lemma \ref{lemma:gauss} we can write
\begin{align}
\label{eq:gauss_lowerTR}
&R_{lo}(\mb{D}) \ge \wh{R}_{lo}(\mb{D})
\end{align}
where 
\begin{align*}
&\wh{R}_{lo}(\mb{D}) = \inf_{\wh{C}_l(\mb{D})} \max \{\wh{R}_{lo1}, \wh{R}_{lo2} \} 
\end{align*}
and
$\wh{C}_l(\mb{D}) = \{(\mb{W},\mb{U},\mb{V}) \in C_l(\mb{D}) | (\mb{W},\mb{U},\mb{V}) \mbox{ jointly Gaussian with } (\mb{X},\mb{Y},\mb{Y_1},\mb{Y_2})\}$.

The following lemmas show that without loss of optimality we can add the conditions $ {\mb{U}} \perp (\mb{X})_{l_1}|(\mb{W},\mb{Y_1}), \ {\mb{V}} \perp [\mb{X}]_{l_2}|(\mb{W}, \mb{Y_2})$,
 and 
$K_{\mb{X}|\mb{W},\mb{Y_1}},K_{\mb{X}|\mb{W},\mb{U},\mb{Y_1}},K_{\mb{X}|\mb{W},\mb{V},\mb{Y_2}}$ are diagonal matrices to   $\wh{C}_l(\mb{D})$.
\begin{lemma}
\label{lemma:diag}
One can add the constraint that  $K_{\mb{X}|\mb{W},\mb{U},\mb{Y_1}},K_{\mb{X}|\mb{W},\mb{V},\mb{Y_2}}$  are diagonal matrices to   $\wh{C}_l(\mb{D})$  without increasing the optimal value,  $\wh{R}_{lo}(\mb{D})$.
\end{lemma} 
 
 \begin{proof}
 Note that for each feasible $(\mb{W},\mb{U},\mb{V})$ in $\wh{C}_l(\mb{D})$,  we can find a $(\mb{W}',\mb{U}',\mb{V}')$ jointly Gaussian with $(\mb{X},\mb{Y},\mb{Y_1},\mb{Y_2})$ and $ (\mb{W}',\mb{U}',\mb{V}') \lra \mb{X} \lra (\mb{Y_1},\mb{Y_2},\mb{Y})$ such that 
 \begin{align*}
 &K_{\mb{X}|\mb{W}',\mb{U}',\mb{Y_1}} = (K_{\mb{X}|\mb{W},\mb{U},\mb{Y_1}})_{diag} 
 \\
 &K_{\mb{X}|\mb{W}',\mb{V}',\mb{Y_2}} = (K_{\mb{X}|\mb{W},\mb{V},\mb{Y_2}})_{diag} 
 \end{align*}
 since $K_{\mb{X}|\mb{W}',\mb{U}',\mb{Y_i}} \preceq D_iI \preceq K_{\mb{X}|\mb{Y_i}}$ for $i \in \{1,2\}$. Also notice that $(\mb{W}',\mb{U}',\mb{V}')$ satisfies the corresponding distortion constraints. Lastly we need to check that  $(K_{\mb{X}|\mb{W}',\mb{U}',\mb{Y}})_{diag} \succeq (K_{\mb{X}|\mb{W},\mb{U},\mb{Y}})_{diag}$ and $(K_{\mb{X}|\mb{W}',\mb{V}',\mb{Y}})_{diag} \succeq (K_{\mb{X}|\mb{W},\mb{V},\mb{Y}})_{diag}$. 
Since 
\begin{align*}
K_{\mb{X}|\mb{W'},\mb{U'},\mb{Y}} &= \left[K_{\mb{X}|\mb{W'},\mb{U'},\mb{Y_1}}^{-1} +\left(\begin{array}{c c}  A&\mf{0}\\ \mf{0}&\mf{0}\end{array} \right) \right]^{-1}
\\
&= \left[((K_{\mb{X}|\mb{W},\mb{U},\mb{Y_1}})_{diag})^{-1} +\left(\begin{array}{c c}  A&\mf{0}\\ \mf{0}&\mf{0}\end{array} \right) \right]^{-1}
\end{align*} 
from Lemma \ref{lemma:matrixineq}  we have $K_{\mb{X}|\mb{W}',\mb{U}',\mb{Y}} \succeq (K_{\mb{X}|\mb{W},\mb{U},\mb{Y}})_{diag}$ and similarly $K_{\mb{X}|\mb{W}',\mb{V}',\mb{Y}} \succeq (K_{\mb{X}|\mb{W},\mb{V},\mb{Y}})_{diag}$.
Hence, without loss of optimality we can add the condition that $K_{\mb{X}|\mb{W},\mb{U},\mb{Y_1}},K_{\mb{X}|\mb{W},\mb{V},\mb{Y_2}}$  are diagonal matrices to  $\wh{C}_l(\mb{D})$.
 \end{proof}
 By Lemma \ref{lemma:diag}, we can write
\begin{align}
\label{ineq:after_diag}
&\wh{R}_{lo}(\mb{D}) = \inf_{\wh{\wh{C}}_l(\mb{D})} \max \{\wh{R}_{lo1}, \wh{R}_{lo2} \} 
\end{align}
where
$\wh{\wh{C}}_l(\mb{D}) = \{(\mb{W},\mb{U},\mb{V}) \in \wh{C}_l(\mb{D}) | K_{\mb{X}|\mb{W},\mb{U},\mb{Y_1}},K_{\mb{X}|\mb{W},\mb{V},\mb{Y_2}}  \mbox{ are diagonal}\}$.
 
 \begin{lemma}
 \label{lemma:indp}
 One may add the constraints 
 \begin{align}
 \label{lb:cond_mc1}
 & \mb{U}  \perp (\mb{X})_{l_1}| \mb{W},\mb{Y_1},
 \\
 \label{lb:cond_mc2}
 & \mb{V} \perp [\mb{X}]_{l_2}| \mb{W},\mb{Y_2},
 \\
 \label{lb:cond_diag}
 &K_{\mb{X}|\mb{W_G},\mb{Y_1}}, K_{\mb{X}|\mb{W_G},\mb{Y_2}}  \mbox{ are diagonal matrices}
 \end{align}
to the optimization set $\wh{\wh{C}}_l(\mb{D})$  without increasing the optimal value, $\wh{R}_{lo}(\mb{D})$.
\end{lemma} 
 
 \begin{proof}
 Let $(\mb{W},\mb{U},\mb{V})$ be feasible for $\wh{R}_{lo}(\mb{D})$, i.e $(\mb{W},\mb{U},\mb{V}) \in \wh{\wh{C}}_l(\mb{D})$. From these, we shall construct $(\wt{\mb{W}}, \wt{\mb{U}}, \wt{\mb{V}})$ that are feasible for
$\wh{R}_{lo}(\mb{D})$  and also satisfy the conditions in (\ref{lb:cond_mc1}), (\ref{lb:cond_mc2}), (\ref{lb:cond_diag}) and for which the objective is only lower. 

First suppose that $d_2 \le d_1$. Then note that 
\begin{align*}
\left(\begin{array}{c c}  (K_{\mb{X}|\mb{W},\mb{U},\mb{Y_1}})^{-1}_{l_1}  &\mf{0}\\ \mf{0}&[K_{\mb{X}|\mb{W},\mb{V},\mb{Y_2}}]^{-1}_{l_2}-B\end{array} \right)
&\succeq
 \left(\begin{array}{c c}  d_1^{-1}I  &\mf{0}\\ \mf{0} &d_2^{-1}I -B \end{array} \right)
 \\
&\succeq d_1^{-1}I
 \\
&\succeq K^{-1}_{\mb{X}|\mb{Y_1}}.
\end{align*}

Then we may choose $\wt{\mb{W}}$ such that 
\begin{align}
&K^{-1}_{\mb{X}|\wt{\mb{W}},\mb{Y_1}}  
\label{ineq_w2}
= \left(\begin{array}{c c}  (K_{\mb{X}|\mb{W},\mb{U},\mb{Y_1}})^{-1}_{l_1}  &\mf{0}\\ \mf{0}&[K_{\mb{X}|\mb{W},\mb{V},\mb{Y_2}}]^{-1}_{l_2}-B\end{array} \right)
\end{align}
in which case we have
\begin{align}
&K^{-1}_{\mb{X}|\wt{\mb{W}},\mb{Y_2}} =K^{-1}_{\mb{X}|\wt{\mb{W}},\mb{Y_1}} + K 
\label{ineq_w1}
= \left(\begin{array}{c c}  (K_{\mb{X}|\mb{W},\mb{U},\mb{Y_1}})^{-1}_{l_1} + A &\mf{0}\\ \mf{0}&[K_{\mb{X}|\mb{W},\mb{V},\mb{Y_2}}]^{-1}_{l_2}\end{array} \right).
\end{align}

Likewise, if $d_1 < d_2$, we have
\begin{align*}
 \left(\begin{array}{c c}  (K_{\mb{X}|\mb{W},\mb{U},\mb{Y_1}})^{-1}_{l_1} + A &\mf{0}\\ \mf{0}&[K_{\mb{X}|\mb{W},\mb{V},\mb{Y_2}}]^{-1}_{l_2}\end{array} \right)
 &\succeq
\left(\begin{array}{c c}  d_1^{-1}I + A &\mf{0}\\ \mf{0}& d_2^{-1}I\end{array} \right)
\\
&\succeq d_2^{-1}I 
\\
&\succeq K^{-1}_{\mb{X}|\mb{Y_2}}.
\end{align*}
Hence we may choose $\wt{\mb{W}}$ such that
\begin{align*}
&K^{-1}_{\mb{X}|\wt{\mb{W}},\mb{Y_2}} 
= \left(\begin{array}{c c}  (K_{\mb{X}|\mb{W},\mb{U},\mb{Y_1}})^{-1}_{l_1} + A &\mf{0}\\ \mf{0}&[K_{\mb{X}|\mb{W},\mb{V},\mb{Y_2}}]^{-1}_{l_2}\end{array} \right)
\end{align*}
in which case
\begin{align*}
&K^{-1}_{\mb{X}|\wt{\mb{W}},\mb{Y_1}} =K^{-1}_{\mb{X}|\wt{\mb{W}},\mb{Y_2}} - K 
= \left(\begin{array}{c c}  (K_{\mb{X}|\mb{W},\mb{U},\mb{Y_1}})^{-1}_{l_1}  &\mf{0}\\ \mf{0}&[K_{\mb{X}|\mb{W},\mb{V},\mb{Y_2}}]^{-1}_{l_2}-B\end{array} \right).
\end{align*}
Thus either way, we may choose $\wt{\mb{W}}$ such that (\ref{ineq_w2}) and (\ref{ineq_w1}) hold, 
and so  $K_{\mb{X}|\wt{\mb{W}},\mb{Y_1}}$ and $K_{\mb{X}|\wt{\mb{W}},\mb{Y_2}}$ are both diagonal.

Next we choose $\wt{\mb{U}}$ and $\wt{\mb{V}}$ such that $(\wt{\mb{W}}, \wt{\mb{U}}, \wt{\mb{V}}) \lra \mb{X} \lra (\mb{Y},\mb{Y_1},\mb{Y_2})$
and 
 \begin{align}
K_{\mb{X}|\wt{\mb{W}},\wt{\mb{U}},\mb{Y_1}} 
& = \left(\begin{array}{c c}  (K_{\mb{X}|\wt{\mb{W}},\mb{Y_1}})_{l_1}  &\mf{0}\\ \mf{0}&\min\{ [K_{\mb{X}|\wt{\mb{W}},\mb{Y_1}}]_{l_2}, [K_{\mb{X}|\mb{W},\mb{U},\mb{Y_1}}]_{l_2}\} \end{array} \right)
\notag \\
\label{lb:ineq_u}
& = \left(\begin{array}{c c}  (K_{\mb{X}|\mb{W},\mb{U},\mb{Y_1}})_{l_1}  &\mf{0}\\ \mf{0}&\min\{ [K_{\mb{X}|\wt{\mb{W}},\mb{Y_1}}]_{l_2}, [K_{\mb{X}|\mb{W},\mb{U},\mb{Y_1}}]_{l_2}\} \end{array} \right)
\end{align}
and
\begin{align}
K_{\mb{X}|\wt{\mb{W}},\wt{\mb{V}},\mb{Y_2}} 
&=\left(\begin{array}{c c}  \min\{ (K_{\mb{X}|\wt{\mb{W}},\mb{Y_2}})_{l_1}, (K_{\mb{X}|\mb{W},\mb{V},\mb{Y_2}})_{l_1}\}  &\mf{0}\\ \mf{0}&[K_{\mb{X}|\wt{\mb{W}},\mb{Y_2}}]_{l_2}\end{array} \right)
\notag \\
\label{lb:ineq_v}
&= \left(\begin{array}{c c}  \min\{ (K_{\mb{X}|\wt{\mb{W}},\mb{Y_2}})_{l_1}, (K_{\mb{X}|\mb{W},\mb{V},\mb{Y_2}})_{l_1}\}  &\mf{0}\\ \mf{0}&[K_{\mb{X}|\mb{W},\mb{V},\mb{Y_2}}]_{l_2}\end{array} \right).
\end{align}
Evidently we have  $\wt{\mb{U}}  \perp (\mb{X})_{l_1}| \wt{\mb{W}},\mb{Y_1}$ and
 $ \wt{\mb{V}} \perp [\mb{X}]_{l_2}| \wt{\mb{W}},\mb{Y_2}$, and $(\wt{\mb{W}}, \wt{\mb{U}}, \wt{\mb{V}})$ satisfy the distortion constraints.
 
 Finally, from (\ref{ineq_w2}) we have
 \begin{align}
 K^{-1}_{\mb{X}|\wt{\mb{W}},\mb{Y}} &= K^{-1}_{\mb{X}|\wt{\mb{W}},\mb{Y_1}} + 
 \left(\begin{array}{c c}  A  &\mf{0}\\ \mf{0}&\mf{0}\end{array} \right)
 \notag \\
 &= \left(\begin{array}{c c} (K^{-1}_{\mb{X}|\mb{W},\mb{U},\mb{Y_1}})_{l_1} + A &\mf{0}\\ \mf{0}&[K^{-1}_{\mb{X}|\mb{W},\mb{V},\mb{Y_2}}]_{l_2} - B\end{array} \right)
 \notag \\
 \label{lb:ineq_WY}
 & =  \left(\begin{array}{c c} (K^{-1}_{\mb{X}|\mb{W},\mb{U},\mb{Y}})_{l_1}  &\mf{0}\\ \mf{0}&[K^{-1}_{\mb{X}|\mb{W},\mb{V},\mb{Y}}]_{l_2}\end{array} \right).
 \end{align}
 
 Similarly, 
  \begin{align*}
 K^{-1}_{\mb{X}|\wt{\mb{W}},\wt{\mb{U}},\mb{Y}} &= K^{-1}_{\mb{X}|\wt{\mb{W}},\wt{\mb{U}},\mb{Y_1}} + 
 \left(\begin{array}{c c}  A  &\mf{0}\\ \mf{0}&\mf{0}\end{array} \right).
 \end{align*}
 
 Substituting (\ref{lb:ineq_u}) into this equation gives,
 \begin{align}
  \label{lb:ineq_UY}
 K_{\mb{X}|\wt{\mb{W}},\wt{\mb{U}},\mb{Y}} &= \left(\begin{array}{c c} (K_{\mb{X}|\mb{W},\mb{U},\mb{Y}})_{l_1}  &\mf{0}\\ \mf{0}& \min \{[K_{\mb{X}|\mb{W},\mb{U},\mb{Y}}]_{l_2}, [K_{\mb{X}|\mb{W},\mb{V},\mb{Y}}]_{l_2} \} \end{array} \right).
 \end{align}
  
 Likewise,
 \begin{align}
  \label{lb:ineq_VY}
 K_{\mb{X}|\wt{\mb{W}},\wt{\mb{V}},\mb{Y}} &= \left(\begin{array}{c c} \min \{(K_{\mb{X}|\mb{W},\mb{U},\mb{Y}})_{l_1}, (K_{\mb{X}|\mb{W},\mb{V},\mb{Y}})_{l_1} \}  &\mf{0}\\ \mf{0}&[K_{\mb{X}|\mb{W},\mb{V},\mb{Y}}]_{l_2}  \end{array} \right).
 \end{align}
 
 From (\ref{lb:ineq_u}), (\ref{lb:ineq_v}), (\ref{lb:ineq_UY}) and (\ref{lb:ineq_VY}), we see that the objective for $(\wt{\mb{W}}, \wt{\mb{U}}, \wt{\mb{V}})$ is equal to the objective for $(\mb{W}, \mb{U}, \mb{V})$. 
\end{proof}

By Lemma \ref{lemma:indp},  we can conclude that 
$\wh{R}_{lo}(\mb{D})$  is equal to
$R^{\tr}_{u}(\mb{D})$.

\section{Concluding Remarks}
\label{sec:conclusion}

Recall that we used the Enhanced-Enhancement lower bound to prove the
converse for the $\RDmse$ the $\RDSI$ problems, while for
the $\RDTr$ problem we used the $\mLB$.
It appears that the other lower bounds are in fact insufficient
for the $\RDTr$ problem.
\begin{conjecture}
There exists an instance of the $\RDTr$ such that the Minimax lower bound is strictly greater than the Maximin lower bound (and hence the Enhanced-Enhancement lower bound and the \DSM).
\end{conjecture}

To support this conjecture, one can apply the same arguments in the 
proof of Proposition \ref{prop:minmaxtrace} to the each 
minimization in the $\MLB$ separately. This way we obtain a lower 
bound, which is the same as in (\ref{eq:minimax_lbtr}) except 
that the minimization and maximization are swapped.
Consider the case where the vectors $\mb{X}$, 
$\mb{Y_1}$, $\mb{Y_2}$ are bivariate Gaussian random vectors such that 
 $$ K_{\mb{X}|\mb{Y_1}} = \left( \begin{array}{c c} \frac{4}{9} & 0 \\ 0 &\frac{4}{9} \end{array}\right),
\ \ \ 
  K_{\mb{X}|\mb{Y_2}} = \left( \begin{array}{c c} \frac{4}{17} & 0 \\  0 & \frac{4}{5} \end{array}\right)$$
and the distortion constraints are $d_1 = d_2 = 0.15$.
When we use \textit{CVX}, a package for solving convex programs
\cite{cvx,cvx2}, and the \textit{sqp} function of Octave \cite{octave} to 
solve for the minimum rate using Theorem~\ref{thm:trace} 
we get a solution of $1.7808784$ 
while we get $1.7802127$ from both solvers 
when we swap the $\min$ and $\max$ in 
(\ref{eq:minimax_lbtr}). Thus it appears
that there are instances for which the added strength provided
by the $\mLB$ is necessary.

\section{Acknowledgment}
This work was supported by Intel, Cisco, and Verizon under the Video-Aware Wireless Networks (VAWN) program and by the National Science Foundation under grants CCF-1117128 and CCF-1218578.
\begin{appendices}
\section{ }
\label{app:ELB}

The aim of this appendix is to prove the following lemma.
\begin{lemma}[Gaussian Variance-Drop Lemma]
\label{lemma:inequality}
 Let $(\mb{W},\mb{W_G},\mb{X},\wt{\mb{Z}},\mb{Z})$ be random vectors
such that $(\mb{W_G},\mb{X},\wt{\mb{Z}},\mb{Z})$ is jointly Gaussian, $(\mb{W},\mb{W_G}) \lra \mb{X} \lra \wt{\mb{Z}} \lra\mb{Z}$
and $K_{\mb{X}|\mb{W},\wt{\mb{Z}}} \succ 0$.
If $K_{\mb{X}|\mb{W_G},\mb{Z}} = K_{\mb{X}|\mb{W},\mb{Z}}$ then $K_{\mb{X}|\mb{W},\wt{\mb{Z}}} \preceq K_{\mb{X}|\mb{W_G},\wt{\mb{Z}}}$. Also,  if $K_{\mb{X}|\mb{W},\wt{\mb{Z}}} = K_{\mb{X}|\mb{W_G},\wt{\mb{Z}}}$ then $K_{\mb{X}|\mb{W_G},\mb{Z}} \preceq K_{\mb{X}|\mb{W},\mb{Z}}$.
\end{lemma}
This lemma can be interpreted as follows. We view $\mb{X}$ as 
an underlying source of interest and $\mb{W}$, $\mb{W}_G$, $\tilde{\mb{Z}}$,
and $\mb{Z}$ as ``noisy observations'' of $\mb{X}$. All except possibly
$\mb{W}$ are jointly Gaussian. If $(\mb{W},\mb{Z})$ and $(\mb{W}_G,\mb{Z})$
are equally-good observations, in terms of their error covariance matrix,
then $(\mb{W},\tilde{\mb{Z}})$ can only be better than
$(\mb{W}_G,\tilde{\mb{Z}})$. That is, replacing $\mb{Z}$ with 
$\tilde{\mb{Z}}$ results in a ``variance drop,'' and this drop
is smallest in the Gaussian case.

To prove this result we will make use of the following technical lemma.
\begin{lemma}
\label{lemma:hatZ}
 Let $(\mb{X},\wt{\mb{Z}},\mb{Z})$ be jointly Gaussian random vectors
such that  $\mb{X} \lra \wt{\mb{Z}} \lra{\mb{Z}}$
and $K_{\mb{X}|\wt{\mb{Z}}} \succ 0$.
We can form a $\wh{\mb{Z}}$ such that  $(\mb{X},\wt{\mb{Z}},\mb{Z},\wh{\mb{Z}})$ is jointly Gaussian,  $\wh{\mb{Z}} \lra \mb{X} \lra\mb{Z}$, and  $E[\mb{X} | \mb{Z}, \wh{\mb{Z}}] = E[\mb{X} | \mb{Z},\wt{\mb{Z}}]$ almost surely.
\end{lemma}

\begin{proof}
Given such $(\mb{X},\wt{\mb{Z}},\mb{Z})$, we can create a $\bar{\mb{Z}}$ such that $\bar{\mb{Z}} = A_{\bar{z}}\mb{X} + \mb{N_{\bar{z}}}$ where $\mb{N_{\bar{z}}}$ is Gaussian, independent of $(\mb{X},\mb{Z})$ and $K_{\mb{X}|\mb{Z}, \mb{\bar{Z}}} = K_{\mb{X}|\mb{Z}, \mb{\wt{Z}}} = K_{\mb{X}| \mb{\wt{Z}}}$. Since $(\mb{X}, \mb{Z}, \mb{\bar{Z}}, E[\mb{X}| \mb{Z}, \mb{\bar{Z}}])$ are jointly Gaussian, we can write
\begin{align*}
\bar{\mb{Z}} = B \left(\begin{array}{c}  \mb{X}\\  \mb{Z} \\
E[\mb{X}| \mb{Z}, \mb{\bar{Z}}] \end{array} \right)
+ \mb{N_{\bar{z}}}',
\end{align*}
for some matrix $B$ where $\mb{N_{\bar{z}}}'$ is independent of $(\mb{X}, \mb{Z},  E[\mb{X}| \mb{Z}, \mb{\bar{Z}}])$ and  Gaussian with some covariance matrix $K_{\mb{N_{\bar{z}}}'}$.

Observe that the orthogonality principle and the equation $K_{\mb{X}|\mb{Z}, \mb{\bar{Z}}} = K_{\mb{X}|\mb{Z}, \mb{\wt{Z}}}$ together imply that
\begin{align}
\label{lemma:nvar_eq1}
& K_{E[\mb{X}|\mb{Z}, \mb{\bar{Z}}]} = K_{E[\mb{X}|\mb{Z}, \mb{\wt{Z}}]}.
\end{align}
Orthogonality also implies that 
$K_{\mb{X}E[\mb{X}|\mb{Z}, \mb{\bar{Z}}]} = K_{E[\mb{X}|\mb{Z}, \mb{\bar{Z}}]}$ and $K_{\mb{X}E[\mb{X}|\mb{Z}, \mb{\wt{Z}}]} = K_{E[\mb{X}|\mb{Z}, \mb{\wt{Z}}]}$. Hence,
\begin{align}
\label{lemma:nvar_eq2}
&K_{\mb{X}E[\mb{X}|\mb{Z}, \mb{\bar{Z}}]} = K_{\mb{X}E[\mb{X}|\mb{Z}, \mb{\wt{Z}}]}.
\end{align}
Likewise, orthogonality implies that $K_{E[\mb{X}|\mb{Z}, \mb{\wt{Z}}]\mb{Z}} = K_{\mb{X}\mb{Z}}$ and $K_{E[\mb{X}|\mb{Z}, \mb{\bar{Z}}]\mb{Z}} = K_{\mb{X}\mb{Z}}$. Thus, 
\begin{align}
\label{lemma:nvar_eq3}
K_{E[\mb{X}|\mb{Z}, \mb{\wt{Z}}]\mb{Z}} = K_{E[\mb{X}|\mb{Z}, \mb{\bar{Z}}]\mb{Z}}.
\end{align}

Then (\ref{lemma:nvar_eq1}), (\ref{lemma:nvar_eq2}), and (\ref{lemma:nvar_eq3}) imply that $(\mb{X},\mb{Z},E[\mb{X}|\mb{Z}, \mb{\bar{Z}}])$ and $(\mb{X},\mb{Z},E[\mb{X}|\mb{Z}, \mb{\wt{Z}}])$ are equal in distribution. Now given $(\mb{X}, \mb{\wt{Z}}, \mb{Z})$, create $\mb{\wh{Z}}$ via
\begin{align*}
\mb{\wh{Z}} = B \left(\begin{array}{c}  \mb{X}\\  \mb{Z} \\
E[\mb{X}| \mb{Z}, \mb{\wt{Z}}] \end{array} \right)
+ \mb{N_{\wh{z}}}',
\end{align*}
where $\mb{N_{\wh{z}}}'$ is Gaussian with covariance matrix $K_{\mb{N_{\bar{z}}}'}$ and is independent of $(\mb{X}, E[\mb{X}|\mb{Z}, \mb{\wt{Z}}], \mb{Z})$. Then,
\begin{align*}
(\mb{X},\mb{Z},\mb{\wh{Z}},E[\mb{X}|\mb{Z}, \mb{\wt{Z}}]) = (\mb{X},\mb{Z},\mb{\bar{Z}},E[\mb{X}|\mb{Z}, \mb{\bar{Z}}]), \mbox{in distribution}
\end{align*}
and so  $\wh{\mb{Z}} \lra \mb{X} \lra\mb{Z}$, and  $E[\mb{X} | \mb{Z}, \wh{\mb{Z}}] = E[\mb{X} | \mb{Z},\wt{\mb{Z}}]$ almost surely.
\end{proof}

\begin{proof}[Proof of Lemma~\ref{lemma:inequality}]
Let $(\mb{W},\mb{W_G},\mb{X},\wt{\mb{Z}}, \mb{Z})$ be as in the statement.
Then by Lemma \ref{lemma:hatZ}, we can form a random vector $\wh{\mb{Z}} = A_{\wh{z}}\mb{X} + \mb{N_{\wh{z}}}$, where $\mb{N_{\wh{z}}}$ is independent of $(\mb{X},\mb{Z})$,  such that  $\wh{\mb{Z}} \lra \mb{X}\lra\mb{Z}$, $K_{\mb{X}|\mb{Z},\wh{\mb{Z}}} = K_{\mb{X}|\mb{Z},\wt{\mb{Z}}}=K_{\mb{X}|\wt{\mb{Z}}}$ and $E[\mb{X}|\mb{Z}, \wt{\mb{Z}}] = E[\mb{X}| \mb{Z}, \wh{\mb{Z}}]$ almost surely. Since for any $\mb{W}$ such that $\mb{W} \lra \mb{X} \lra (\wt{\mb{Z}}, \wh{\mb{Z}}, \mb{Z})$ we have $K_{\mb{X}|\mb{W},\wt{\mb{Z}}} = K_{\mb{X}|\mb{W},\mb{Z},\wt{\mb{Z}}} =
K_{\mb{X}|\mb{W},E[\mb{X}|\mb{Z},\wt{\mb{Z}}]} = K_{\mb{X}|\mb{W},E[\mb{X}|\mb{Z},\wh{\mb{Z}}]} = K_{\mb{X}|\mb{W},\mb{Z},\wh{\mb{Z}}}$, it suffices to prove the result for the special case in which $\wt{\mb{Z}} = (\wh{\mb{Z}},\mb{Z})$ so we shall assume that $\wt{\mb{Z}}$ has this form.
 Also, let $\wh{\mb{X}} = E[\mb{X}|\mb{W},\mb{Z}]$. We will write the covariance matrix of the best linear estimate of $\mb{X}$ using $\wh{\mb{X}}$ and $\wh{\mb{Z}}$ in terms of $K_{\mb{X}|\mb{W},\mb{Z}}$ and $K_{\mb{X}|\wh{\mb{Z}}}$ by applying the procedure of \cite{Wang}.  Then we can write
\begin{align*}
&K(\mb{X},\wh{\mb{X}},\wh{\mb{Z}})=
 \left( \begin{array}{c c c} K_\mb{X} &K_{\wh{\mb{X}}} & K_\mb{X}A^{T}_{\wh{z}}\\ K_{\wh{\mb{X}}}& K_{\wh{\mb{X}}} &K_{\wh{\mb{X}}}A^{T}_{\wh{z}} \\ A_{\wh{z}}K_\mb{X}&A_{\wh{z}}K_{\wh{\mb{X}}} &A_{\wh{z}}K_\mb{X}A^{T}_{\wh{z}} + K_{\mb{N_{\wh{z}}}}\end{array} \right)
\end{align*}
where
$K_{\wh{\mb{X}}}= (K_{\mb{X}}-K_{\mb{X}|\mb{W},\mb{Z}})$. Note that $K_{\wh{\mb{X}}}$ may not be invertible, meaning that some of the elements of $\wh{\mb{X}}$ can be determined as linear combinations of  others. Thus it is enough to consider only the components of $\wh{\mb{X}}$ or linear combinations of them, denoted by $\bar{\mb{X}} = Q\wh{\mb{X}}$, such that the resulting covariance matrix, denoted as $K_{\bar{\mb{X}}} = QK_{\wh{\mb{X}}}Q^T$, is invertible. Then we can write,
\begin{align*}
&K(\mb{X},\bar{\mb{X}},\wh{\mb{Z}})=
 \left( \begin{array}{c c c} K_\mb{X} &K_{\wh{\mb{X}}}Q^T & K_\mb{X}A^{T}_{\wh{z}}\\ QK_{\wh{\mb{X}}}& K_{\bar{\mb{X}}} &QK_{\wh{\mb{X}}}A^{T}_{\wh{z}} \\ A_{\wh{z}}K_\mb{X}&A_{\wh{z}}K_{\wh{\mb{X}}}Q^T &A_{\wh{z}}K_\mb{X}A^{T}_{\wh{z}} + K_{\mb{N_{\wh{z}}}}\end{array} \right).
\end{align*}
The covariance matrix of a linear estimate of $\mb{X}$ using $\wh{\mb{X}}$ and $\wh{\mb{Z}}$ is
 \begin{align*}
&K_{(\mb{X}|\wh{\mb{X}},\wh{\mb{Z}})_L}= K_{(\mb{X}|\bar{\mb{X}},\wh{\mb{Z}})_L} =
K_\mb{X} - \left( \begin{array}{c c}  K_{\wh{\mb{X}}}Q^T & K_\mb{X}A^{T}_{\wh{z}} \end{array} \right)C^{-1}\left( \begin{array}{c}  QK_{\wh{\mb{X}}}\\ A_{\wh{z}}K_\mb{X}\end{array} \right)
\end{align*}
where 
\begin{align*}
C=\left( \begin{array}{c c } K_{\bar{\mb{X}}} &QK_{\wh{\mb{X}}}A^{T}_{\wh{z}} \\ A_{\wh{z}}K_{\wh{\mb{X}}}Q^T &A_{\wh{z}}K_\mb{X}A^{T}_{\wh{z}} + K_{\mb{N_{\wh{z}}}}\end{array} \right).
\end{align*}

By the matrix inversion lemma we have
\begin{align*}
&K^{-1}_{(\mb{X}|\wh{\mb{X}},\wh{\mb{Z}})_L}=K^{-1}_\mb{X} +
 K^{-1}_\mb{X} \left( \begin{array}{c c}  K_{\wh{\mb{X}}}Q^T & K_\mb{X}A^{T}_{\wh{z}} \end{array} \right)E^{-1}\left( \begin{array}{c}  QK_{\wh{\mb{X}}}\\ A_{\wh{z}}K_\mb{X}\end{array} \right)K^{-1}_\mb{X} 
\end{align*}
where 
\begin{align*}
E &= C - \left(\begin{array}{c}  QK_{\wh{\mb{X}}}\\ A_{\wh{z}}K_\mb{X}\end{array} \right)K^{-1}_\mb{X} \left( \begin{array}{c c}  K_{\wh{\mb{X}}}Q^T& K_\mb{X}A^{T}_{\wh{z}} \end{array} \right)
 \\
&= C - \left(\begin{array}{c}  Q(I-K_{\mb{X}|\mb{W},\mb{Z}}K^{-1}_\mb{X})\\ A_{\wh{z}}\end{array} \right) \left( \begin{array}{c c}  K_{\wh{\mb{X}}}Q^T & K_\mb{X}A^{T}_{\wh{z}} \end{array} \right) 
\\
&= C- \left(\begin{array}{c c}  Q(K_{\wh{\mb{X}}}-K_{\mb{X}|\mb{W},\mb{Z}}K^{-1}_\mb{X}K_{\wh{\mb{X}}})Q^T & Q\wh{K}_xA^{T}_{\wh{z}}  \\ A_{\wh{z}}K_{\wh{\mb{X}}}Q^T & A_{\wh{z}}K_{\mb{X}}A^{T}_{\wh{z}}\end{array} \right) 
\\
&= \left( \begin{array}{c c } K_{\bar{\mb{X}}} &QK_{\wh{\mb{X}}}A^{T}_{\wh{z}} \\ A_{\wh{z}}K_{\wh{\mb{X}}}Q^T &A_{\wh{z}}K_\mb{X}A^{T}_{\wh{z}} + K_{\mb{N_{\wh{z}}}}\end{array} \right)
- \left(\begin{array}{c c}  Q(K_{\wh{\mb{X}}}-K_{\mb{X}|\mb{W},\mb{Z}}K^{-1}_\mb{X}K_{\wh{\mb{X}}})Q^T & QK_{\wh{\mb{X}}}A^{T}_{\wh{z}}  \\ A_{\wh{z}}K_{\wh{\mb{X}}}Q^T & A_{\wh{z}}K_{\mb{X}}A^{T}_{\wh{z}}\end{array} \right) 
 \\
&= \left( \begin{array}{c c } Q(K_{\mb{X}|\mb{W},\mb{Z}} -K_{\mb{X}|\mb{W},\mb{Z}}K^{-1}_\mb{X}K_{\mb{X}|\mb{W},\mb{Z}})Q^T &0 \\ 0&K_{\mb{N_{\wh{z}}}}\end{array} \right).
\end{align*}

Then
\begin{align}
K^{-1}_{(\mb{X}|\wh{\mb{X}},\wh{\mb{Z}})_L}&=K^{-1}_\mb{X} +
 K^{-1}_\mb{X} \left( \begin{array}{c c}  K_{\wh{\mb{X}}}Q^T & K_\mb{X}A^{T}_{\wh{z}} \end{array} \right) 
\left( \begin{array}{c c } K_{(\bar{\mb{X}}|\mb{X})_L}^{-1}&0 \\ 0&K^{-1}_{\wh{\mb{Z}}|\mb{X}}\end{array} \right) 
\left( \begin{array}{c}  QK_{\wh{\mb{X}}}\\ A_{\wh{z}}K_\mb{X}\end{array} \right)K^{-1}_\mb{X}  
\notag \\
&=K^{-1}_\mb{X} + K^{-1}_\mb{X} \left( \begin{array}{c c}  K_{\wh{\mb{X}}}Q^TK_{(\bar{\mb{X}}|\mb{X})_L}^{-1}& K_\mb{X}A^{T}_{\wh{z}}K^{-1}_{\wh{\mb{Z}}|\mb{X}} \end{array} \right) 
\left( \begin{array}{c}  QK_{\wh{\mb{X}}}\\ A_{\wh{z}}K_\mb{X}\end{array} \right)K^{-1}_\mb{X}  
\notag \\
&= K^{-1}_\mb{X} + K^{-1}_\mb{X} \left( \begin{array}{c c}  K_{\wh{\mb{X}}}Q^TK_{(\bar{\mb{X}}|\mb{X})_L}^{-1}QK_{\wh{\mb{X}}} & K_\mb{X}A^{T}_{\wh{z}}K^{-1}_{\wh{\mb{Z}}|\mb{X}}A_{\wh{z}}K_\mb{X} \end{array} \right) K^{-1}_\mb{X}
 \notag \\
&= K^{-1}_{(\mb{X}|\bar{\mb{X}})_L} + K_{\mb{X}|\wh{\mb{Z}}} - K^{-1}_\mb{X}, \mbox{ by matrix inversion lemma}
 \notag \\
&= K^{-1}_{\mb{X}|\bar{\mb{X}}} + K_{\mb{X}|\wh{\mb{Z}}} - K^{-1}_\mb{X}
\notag \\
&= K^{-1}_{\mb{X}|\wh{\mb{X}}} + K_{\mb{X}|\wh{\mb{Z}}} - K^{-1}_\mb{X}
\notag \\
&= K^{-1}_{\mb{X}|\mb{W},\mb{Z}} + K_{\mb{X}|\wh{\mb{Z}}} - K^{-1}_\mb{X}.
\notag
\end{align}

Hence we have 
\begin{align}
K^{-1}_{(\mb{X}|\wh{\mb{X}},\wh{\mb{Z}})_L}&=
 K^{-1}_{\mb{X}|\mb{W},\mb{Z}}+K^{-1}_{\mb{X}|\wh{\mb{Z}}} -K^{-1}_{\mb{X}} 
+ K^{-1}_{\mb{X}|\mb{W_G},Z} -K^{-1}_{\mb{X}|\mb{W_G},\mb{Z}} \notag \\
\label{cor}
&=  K^{-1}_{\mb{X}|\mb{W_G},\mb{Z},\wh{\mb{Z}}} +K^{-1}_{\mb{X}|\mb{W},\mb{Z}} - K^{-1}_{\mb{X}|\mb{W_G},\mb{Z}}.
\end{align}

Note that $K_{\mb{X}|\mb{W},\mb{Z},\wh{\mb{Z}}} \preceq K_{({\mb{X}|\wh{\mb{X}},\wh{\mb{Z}}})_L}$ so  $K^{-1}_{\mb{X}|\mb{W},\mb{Z},\wh{\mb{Z}}} \succeq K^{-1}_{({\mb{X}|\wh{\mb{X}},\wh{\mb{Z}}})_L}$.  Then,  from (\ref{cor}) we have 
\begin{align}
\label{lemmainequality}
K^{-1}_{\mb{X}|\mb{W},\mb{Z},\wh{\mb{Z}}} \succeq K^{-1}_{\mb{X}|\mb{W_G},\mb{Z},\wh{\mb{Z}}} +K^{-1}_{\mb{X}|\mb{W},\mb{Z}} - K^{-1}_{\mb{X}|\mb{W_G},\mb{Z}}.
\end{align}

Thus, by (\ref{lemmainequality}) if  $K_{\mb{X}|\mb{W_G},\mb{Z}} = K_{\mb{X}|\mb{W},\mb{Z}}$ then $K_{\mb{X}|\mb{W},\wt{\mb{Z}}} \preceq K_{\mb{X}|\mb{W_G},\wt{\mb{Z}}}$ and if $K_{\mb{X}|\mb{W},\wt{\mb{Z}}} = K_{\mb{X}|\mb{W_G},\wt{\mb{Z}}}$ then $K_{\mb{X}|\mb{W_G},\mb{Z}} \preceq K_{\mb{X}|\mb{W},\mb{Z}}$.
\end{proof}
Lemma \ref{lemma:inequality} leads us to the following corollary.
\begin{corollary}
\label{corr}
Let $(\mb{W},\mb{X},\mb{Z},\wt{\mb{Z}})$ be random vectors such that
$\mb{X}$, $\mb{Z}$, and $\wt{\mb{Z}}$ are jointly Gaussian,
$\mb{W} \lra \mb{X} \lra \wt{\mb{Z}}\lra \mb{Z}$ and
$K_{\mb{X}|\mb{W},\wt{\mb{Z}}} \succ 0$.
Also, let $\wt{D} = (D^{-1} + K_{\mb{X}|\wt{\mb{Z}}}^{-1} -K_{\mb{X}|\mb{Z}}^{-1})^{-1}$. If $K_{\mb{X}|\mb{W},\mb{Z}}= D$ then $K_{\mb{X}|\mb{W},\wt{\mb{Z}}} \preceq \wt{D}$. 
\end{corollary}
\begin{proof}
We can find $\mb{W_G}$ jointly Gaussian with
$(\mb{X},\wt{\mb{Z}},\mb{Z})$
such that $(\mb{W},\mb{W_G}) \lra \mb{W} \lra
   \wt{\mb{Z}} \lra \mb{Z}$, and 
$K_{\mb{X}|\mb{W_G},\mb{Z}} = K_{\mb{X}|\mb{W},\mb{Z}} =D.$
Then
$K_{\mb{X}|\mb{W_G},\wt{\mb{Z}}}^{-1}  = K_{\mb{X}|\mb{W_G},\mb{Z}}^{-1} + K_{\mb{X}|\wt{\mb{Z}}}^{-1}
                                    - K_{\mb{X}|\mb{Z}}^{-1}$ 
      $= K_{\mb{X}|\mb{W},\mb{Z}}^{-1} + K_{\mb{X}|\wt{\mb{Z}}}^{-1} - K_{\mb{X}|\mb{Z}}^{-1}$.
Lemma \ref{lemma:inequality} then implies the result.
\end{proof}
\section{ }
\label{app:ELB_bound}

\begin{proof}[ Proof of Lemma \ref{lemma:gauss_ELB}]
We show that without loss of optimality, the auxiliary random vectors can be chosen to be jointly Gaussian with $(\mb{X},\mb{Y_1},\mb{Y_2},\mb{Y})$ in $\DSM $ and Enhanced-$\DSM$. Let $\mb{Y} \in S_G$, $(\mb{W},\mb{U}, \mb{V}) \in \wt{C}_{l1}$  ($(\mb{W},\mb{U}, \mb{V}) \in \bar{C}_{l1}$  for Enhanced-$\DSM$) be given and  $R_{lo1} = I(\mb{X};\mb{W},\mb{U}|\mb{Y_1}) + I(\mb{X};\mb{V}|\mb{W},\mb{U},\mb{Y})$ as defined before. Note that  without loss of generality we can write $\mb{Y} = A_{\mb{Y}} \mb{X} + B_{\mb{Y}} \mb{Y_1} + \mb{N}_{\mb{Y}}$,
where $\mb{N}_{\mb{Y}}$ is a Gaussian vector that is independent of the pair
$(\mb{X},\mb{Y_1})$. Then
\begin{align*}
R_{lo1} &= h(\mb{X}|\mb{Y_1}) - h(\mb{X}|\mb{W},\mb{U},\mb{Y_1}) + h(\mb{X}|\mb{W},\mb{U},\mb{Y}) 
- h(\mb{X}|\mb{W},\mb{U},\mb{V},\mb{Y}) 
\\
&= h(\mb{X}|\mb{Y_1}) - h(\mb{X}|\mb{W},\mb{U},\mb{Y_1}) + h(\mb{X}|\mb{W},\mb{U},\mb{Y_1},\mb{Y}) 
- h(\mb{X}|\mb{W},\mb{U},\mb{V},\mb{Y_1},\mb{Y}), \text{ since } \mb{X}-\mb{Y}-\mb{Y_1} 
\\
&=  h(\mb{X}|\mb{Y_1}) -I(\mb{X};\mb{Y}|\mb{W},\mb{U},\mb{Y_1}) - h(\mb{X}|\mb{W},\mb{U},\mb{V},\mb{Y_1},\mb{Y}) 
\\
&=  h(\mb{X}|\mb{Y_1}) +h(\mb{Y}|\mb{X},\mb{Y_1}) -h(\mb{Y}|\mb{W},\mb{U},\mb{Y_1}) 
- h(\mb{X}|\mb{W},\mb{U},\mb{V},\mb{Y_1},\mb{Y}) 
\\
&=  h(\mb{X}|\mb{Y_1}) +h(\mb{Y}|\mb{X},\mb{Y_1}) -h(A_{\mb{Y}}\mb{X} +  \mb{N}_{\mb{Y}}|\mb{W},\mb{U},\mb{Y_1}) 
\\
& \quad
- h(\mb{X}|\mb{W},\mb{U},\mb{V},\mb{Y_1},\mb{Y}),  \mbox{ since } \mb{Y} = A_{\mb{Y}} \mb{X} + B_{\mb{Y}} \mb{Y_1} + \mb{N}_{\mb{Y}}
 \\
&\ge \frac{1}{2}\log{\frac{|K_{\mb{X}|\mb{Y_1}}||K_{\mb{Y}|\mb{X},\mb{Y_1}}|}{|K_{A_{\mb{Y}}\mb{X} +\mb{N}_{\mb{Y}}|\mb{W},\mb{U},\mb{Y_1}}||K_{\mb{X}|\mb{W},\mb{U},\mb{V},\mb{Y_1},\mb{Y}}|}}
\end{align*}
where $K_{A_{\mb{Y}}\mb{X} + \mb{N}_{\mb{Y}}|\mb{W},\mb{U},\mb{Y_1}}
=A_{\mb{Y}}K_{\mb{X}|\mb{W},\mb{U},\mb{Y_1}}A^T_{\mb{Y}} + K_{\mb{N}_{\mb{Y}}}$ and equality holds if $(\mb{W},\mb{U},\mb{V})$ is Gaussian. We can find $(\mb{W_G},\mb{U_G})$ that are jointly Gaussian with $ (\mb{X}, \mb{Y_1}, \mb{Y_2}, \mb{Y})$
such that $(\mb{W_G}, \mb{U_G}) \lra \mb{X} \lra \mb{Y} \lra (\mb{Y_1}, \mb{Y_2})$ and 
$K_{\mb{X}|\mb{W_G},\mb{U_G},\mb{Y_1}} = K_{\mb{X}|\mb{W},\mb{U},\mb{Y_1}}$.
Then by Lemma \ref{lemma:inequality},
$K_{\mb{X}|\mb{W_G},\mb{U_G},\mb{Y}} \succeq K_{\mb{X}|\mb{W},\mb{U},\mb{Y}} \succeq K_{\mb{X}|\mb{W},\mb{U},\mb{V},\mb{Y}}$.
Thus we can find a $\mb{V_G}$ that is jointly Gaussian with $(\mb{W_G},\mb{U_G},\mb{X}, \mb{Y_1}, \mb{Y_2}, \mb{Y})$ such that $(\mb{W_G},\mb{U_G},\mb{V_G}) \lra \mb{X} \lra  \mb{Y} \lra (\mb{Y_1}, \mb{Y_2})$
and $K_{\mb{X}|\mb{W_G},\mb{U_G},\mb{V_G},\mb{Y}} = K_{\mb{X}|\mb{W},\mb{U},\mb{V},\mb{Y}}$, giving $(\mb{W_G}, \mb{U_G}, \mb{V_G}) \in \wt{C}_{l1}$, ($(\mb{W_G}, \mb{U_G}, \mb{V_G}) \in \bar{C}_{l1}$ for Enhanced-$\DSM$).
Therefore, one can choose the auxiliary random vectors to be jointly Gaussian with $(\mb{X}, \mb{Y_1}, \mb{Y_2}, \mb{Y})$
without loss of optimality in $R_{lo1}$. The same argument applies to $R_{lo2}$ as well. 
\end{proof}
\section{}
\label{app:convexity}
\begin{lemma}
\label{lemma:convexity}
$R'_{lo}(\mb{D}, \mb{Y})$ is a convex function with respect to $\mb{D}$.
\end{lemma}
\begin{proof}[Proof of Lemma \ref{lemma:convexity}]
To prove the lemma, we use a similar argument to \cite{wyner}.
 Let $\epsilon > 0$ be given. We can find $(\wt{\mb{W}},\wt{\mb{U}},\wt{\mb{V}})$ and $(\wh{\mb{W}},\wh{\mb{U}},\wh{\mb{V}})$ in $C_l(\wt{\mb{D}})$ and $C_l(\wh{\mb{D}})$ respectively such that
\begin{align*}
& \max \{I(\mb{X};\wt{\mb{W}},\wt{\mb{U}}|\mb{Y_1}) + I(\mb{X};\wt{\mb{V}}|\wt{\mb{W}},\wt{\mb{U}},\mb{Y}),  I(\mb{X};\wt{\mb{W}},\wt{\mb{V}}|\mb{Y_2}) + I(\mb{X};\wt{\mb{U}}|\wt{\mb{W}},\wt{\mb{V}},\mb{Y}) \}
 \le  R'_{lo}(\mb{\wt{D}}, \mb{Y}) + \epsilon
 \\
 & \max \{I(\mb{X};\wh{\mb{W}},\wh{\mb{U}}|\mb{Y_1}) + I(\mb{X};\wh{\mb{V}}|\wh{\mb{W}},\wh{\mb{U}},\mb{Y}),  I(\mb{X};\wh{\mb{W}},\wh{\mb{V}}|\mb{Y_2}) + I(\mb{X};\wh{\mb{U}}|\wh{\mb{W}},\wh{\mb{V}},\mb{Y}) \}
 \le  R'_{lo}(\mb{\wh{D}}, \mb{Y}) + \epsilon.
\end{align*}
Now we construct $(\mb{W},\mb{U},\mb{V})$ and show that it is in  $C_l(\lambda\wt{\mb{D}} + (1-\lambda)\wh{\mb{D}})$. Let $T$ be a binary random variable with $P(T=1) = \lambda$ and independent of $(\wt{\mb{W}},\wt{\mb{U}},\wt{\mb{V}},\wh{\mb{W}},\wh{\mb{U}},\wh{\mb{V}}, \mb{X},\mb{Y_1},\mb{Y_2},\mb{\mb{Y}})$. Then we define 
\begin{align*}
&\mb{W} = (\wt{\mb{W}}, T) \mbox{ if } T=1, &\mb{W} = (\wh{\mb{W}}, T) \mbox{ if } T=0,
\\
&\mb{U} = (\wt{\mb{U}}, T) \mbox{ if } T=1, &\mb{U} = (\wh{\mb{U}}, T) \mbox{ if } T=0,
\\
&\mb{V} = (\wt{\mb{V}}, T) \mbox{ if } T=1, &\mb{V} = (\wh{\mb{V}}, T) \mbox{ if } T=0
\end{align*}
and
\begin{align*}
&g_1(\mb{W},\mb{U},\mb{Y_1}) = E[\mb{X}|\wt{\mb{W}},\wt{\mb{U}},\mb{Y_1}] \mbox{ if } T=1, 
&g_1(\mb{W},\mb{U},\mb{Y_1}) = E[\mb{X}|\wh{\mb{W}},\wh{\mb{U}},\mb{Y_1}] \mbox{ if } T=0,
\\
&g_2(\mb{W},\mb{V},\mb{Y_2}) = E[\mb{X}|\wt{\mb{W}},\wt{\mb{V}},\mb{Y_2}] \mbox{ if } T=1, 
&g_2(\mb{W},\mb{V},\mb{Y_2}) = E[\mb{X}|\wh{\mb{W}},\wh{\mb{V}},\mb{Y_2}] \mbox{ if } T=0.
\end{align*}
Note that $K_{\mb{X}| g_1(\mb{W},\mb{U},\mb{Y_1})} = \lambda K_{\mb{X}| \wt{\mb{W}},\wt{\mb{U}},\mb{Y_1}} + (1-\lambda) K_{\mb{X}| \wh{\mb{W}},\wh{\mb{U}},\mb{Y_1}}$ and since $\Gamma_1$ is a linear operator, $ \Gamma_1(K_{\mb{X}| g_1(\mb{W},\mb{U},\mb{Y_1})}) \preceq
\lambda \wt{D}_1 + (1- \lambda)\wh{D}_1$. Similarly,   that 
$K_{\mb{X}| g_2(\mb{W},\mb{V},\mb{Y_2})} = \lambda K_{\mb{X}| \wt{\mb{W}},\wt{\mb{V}},\mb{Y_2}} + (1-\lambda) K_{\mb{X}| \wh{\mb{W}},\wh{\mb{V}},\mb{Y_2}}$ gives $\Gamma_2(K_{\mb{X}| g_2(\mb{W},\mb{V},\mb{Y_2})})\preceq
\lambda \wt{D}_2 + (1- \lambda)\wh{D}_2$. Hence, $(\mb{W},\mb{U},\mb{V}) \in C_l(\lambda \wt{\mb{D}} + (1- \lambda)\wh{\mb{D}})$. We can write
\begin{align*}
&R'_{lo}(\lambda \wt{\mb{D}} + (1- \lambda)\wh{\mb{D}}, \mb{Y})
\\
& \le \max \{I(\mb{X};\mb{W},\mb{U}|\mb{Y_1}) + I(\mb{X};\mb{V}|\mb{W},\mb{U},\mb{Y}),  I(\mb{X};\mb{W},\mb{V}|\mb{Y_2}) + I(\mb{X}; \mb{U}|\mb{W},\mb{V},\mb{Y}) \}
\\
& = \max \{I(\mb{X};\mb{W},\mb{U},T|\mb{Y_1}) + I(\mb{X};\mb{V}|\mb{W},\mb{U},T,\mb{Y}),  I(\mb{X};\mb{W},\mb{V},T|\mb{Y_2}) + I(\mb{X}; \mb{U}|\mb{W},\mb{V},T,\mb{Y}) \}
\\
& = \max \{I(\mb{X};\mb{W},\mb{U}|\mb{Y_1},T) + I(\mb{X};\mb{V}|\mb{W},\mb{U},T,\mb{Y}),  I(\mb{X};\mb{W},\mb{V}|\mb{Y_2},T) + I(\mb{X}; \mb{U}|\mb{W},\mb{V},T,\mb{Y}) \}
\\
& =  \max \{\lambda I(\mb{X};\wt{\mb{W}},\wt{\mb{U}}|\mb{Y_1}) + (1-\lambda) I(\mb{X};\wh{\mb{W}},\wh{\mb{U}}|\mb{Y_1})+ \lambda I(\mb{X};\wt{\mb{V}}|\wt{\mb{W}},\wt{\mb{U}},\mb{Y}) + (1-\lambda)I(\mb{X};\wh{\mb{V}}|\wh{\mb{W}},\wh{\mb{U}},\mb{Y}), 
\\
&\quad \quad \quad \quad \lambda I(\mb{X};\wt{\mb{W}},\wt{\mb{V}}|\mb{Y_2}) + (1-\lambda) I(\mb{X};\wh{\mb{W}},\wh{\mb{V}}|\mb{Y_2})+ \lambda I(\mb{X};\wt{\mb{U}}|\wt{\mb{W}},\wt{\mb{V}},\mb{Y}) + (1-\lambda)I(\mb{X};\wh{\mb{U}}|\wh{\mb{W}},\wh{\mb{V}},\mb{Y}) \}
\\
& \le \lambda\max \{I(\mb{X};\wt{\mb{W}},\wt{\mb{U}}|\mb{Y_1}) + I(\mb{X};\wt{\mb{V}}|\wt{\mb{W}},\wt{\mb{U}},\mb{Y}),  I(\mb{X};\wt{\mb{W}},\wt{\mb{V}}|\mb{Y_2}) + I(\mb{X};\wt{\mb{U}}|\wt{\mb{W}},\wt{\mb{V}},\mb{Y}) \}
\\
& \quad +(1-\lambda)\max \{I(\mb{X};\wh{\mb{W}},\wh{\mb{U}}|\mb{Y_1}) + I(\mb{X};\wh{\mb{V}}|\wh{\mb{W}},\wh{\mb{U}},\mb{Y}),  I(\mb{X};\wh{\mb{W}},\wh{\mb{V}}|\mb{Y_2}) + I(\mb{X};\wh{\mb{U}}|\wh{\mb{W}},\wh{\mb{V}},\mb{Y}) \}
\\
&\le \lambda R'_{lo}(\wt{\mb{D}}, \mb{Y}) + (1-\lambda)R'_{lo}(\wh{\mb{D}}, \mb{Y}) + \epsilon.
\end{align*}
By letting $\epsilon \rightarrow 0$, we conclude that $R'_{lo}(\mb{D}, \mb{Y})$ is a convex function of $\mb{D}$. 
\end{proof}

\section{}
\label{app:matrix_ineq}
\begin{proof}[Proof of Lemma \ref{lemma:matrixineq}]
First we consider  $A \succ 0$. Using the matrix inversion lemma, we can write 
\begin{align*}
([M^{-1} + A]^{-1})_{diag}& = (A^{-1} - A^{-1}[M + A^{-1}]^{-1}A^{-1} )_{diag} 
\\
& = A^{-1} - A^{-1}([M + A^{-1}]^{-1})_{diag}A^{-1}
\\
&\preceq  A^{-1} - A^{-1}[M_{diag} + A^{-1}]^{-1}A^{-1},
\end{align*}
since $(M_{diag})^{-1} \preceq (M^{-1})_{diag}$,  \cite[Theorem 7.7.8]{Horn}. By the matrix inversion lemma, the right hand side of the last inequality is $[(M_{diag})^{-1} + A]^{-1}$.

Now, we consider   $A \succeq 0$. Without loss of generality we can assume that all positive diagonal entries are on the upper left corner of $A$. Hence we can write 
\begin{align*}
&A= \left(\begin{array}{c c}  A_1&\mf{0}\\ \mf{0}&\mf{0}\end{array} \right), 
\end{align*}
where $A_1 \succ 0$ is an $m_1 \times m_1$ matrix, where $m_1 \le m$.
Also we can represent $M$ in terms of block matrices, 
\begin{align*}
M= \left(\begin{array}{c c}  M_1& M_2\\ M^T_2 & M_3 \end{array} \right), 
\end{align*}
where $M_1 \succ 0$, is an $m_1 \times m_1$ matrix, and $M_3 \succ 0$ is an $(m-m_1) \times (m-m_1)$ matrix.
Then we can write inverse of $M$ as
\begin{align*}
M^{-1}= \left(\begin{array}{c c}  \bar{M}_1& \bar{M}_2\\ \bar{M}^T_2 & \bar{M}_3 \end{array} \right), 
\end{align*}
where 
\begin{align*}
&\bar{M}_1 = (M_1 - M_2M^{-1}_3M^T_2)^{-1}
\\
&\bar{M}_3 = (M_3 - M^T_2 M_3^{-1}M_2)^{-1}
\\
&\bar{M}_2 = -M^{-1}_1M_2(M_1 - M_2M^{-1}_3M^T_2)^{-1}.
\end{align*}

Also,
\begin{align*}
[M^{-1} + A]= \left(\begin{array}{c c}  \bar{M}_1 + A_1& \bar{M}_2\\ \bar{M}_2 & \bar{M}_3 \end{array} \right).
\end{align*}

When we take the inverse of $[M^{-1} + A]$ we have,
\begin{align*}
[M^{-1} + A]^{-1}= \left(\begin{array}{c c}  \wt{M}_1 & \wt{M}_2\\ \wt{M}_2 & \wt{M}_3 \end{array} \right).
\end{align*}
where $\wt{M_2}$ is a matrix in terms of $\bar{M}_1$, $\bar{M}_2$, $\bar{M}_3$, $A$ and 
\begin{align*}
&\wt{M}_1 = (\bar{M}_1 + A_1 - \bar{M}_2\bar{M}^{-1}_3\bar{M}^T_2)^{-1}
\\
&\wt{M}_3 = (\bar{M}_3 - \bar{M}^T_2( \bar{M}_1 + A_1)^{-1}\bar{M}_2)^{-1}.
\end{align*}

Since $M_1 =  (\bar{M}_1 - \bar{M}_2\bar{M}^{-1}_3\bar{M}^T_2)^{-1}$  and $M_3 =  (\bar{M}_3 - \bar{M}^T_2\bar{M}_1^{-1}\bar{M}_2)^{-1}$, we can write
\begin{align*}
&\wt{M}_1 = [M^{-1}_1 + A_1]^{-1}
\\
&\wt{M}_3 \preceq M_3.
\end{align*}

Then utilizing the inequalities above we can write
 \begin{align*}
([M^{-1} + A]^{-1})_{diag}&= \left(\begin{array}{c c}  \wt{M}_1 & \wt{M}_2\\ \wt{M}_2 & \wt{M}_3 \end{array} \right)_{diag}
\\
& \preceq \left(\begin{array}{c c}  ([M^{-1}_1 + A_1]^{-1})_{diag} & \mf{0}\\ \mf{0} & (M_3)_{diag} \end{array} \right)
\\
& \preceq 
[(M_{diag})^{-1} + A]^{-1}.
\end{align*}
\end{proof}


\end{appendices}

\bibliographystyle{IEEEtran}
\bibliography{IEEEabrv,references}

\end{document}